\documentclass[journal]{IEEEtran}
\usepackage{cite}  
\usepackage{color}
\usepackage{algorithm} 
\usepackage{algorithmic} 
\usepackage{multirow} 
\usepackage{amsmath} 
\usepackage{xcolor}

\usepackage{graphicx}
\usepackage{epstopdf}

\usepackage{amsfonts}
\usepackage{stfloats}
\usepackage{amsthm,amsmath,amssymb}
\usepackage{mathrsfs}
\usepackage{hyperref}

\newtheorem{lemma}{Lemma}
\newtheorem{proposition}{Proposition}

\ifCLASSINFOpdf
\else
\fi
\begin{document}
%
\title{A Joint Power Splitting, Active and Passive Beamforming Optimization Framework for IRS Assisted MIMO SWIPT System}
%
%
%

\author{Chen He,~\IEEEmembership{Member,~IEEE,}
        Xie Xie, 
        Kun Yang,~\IEEEmembership{Senior Member,~IEEE,}
        and Z. Jane Wang,~\IEEEmembership{Fellow,~IEEE}
\thanks{Chen He and Xie Xie are with the School of Information Science and Technology, Northwest University, Xi'an, 710069, China. Corresponding author: Chen He (email: chenhe@nwu.edu.cn).}
\thanks{Kun Yang is with the School of Information Science and Technology, Northwest University, Xi'an, 710069, China, and he is also with School of Computer Science \& Electronic Engineering, University of Essex, Wivenhoe Park, Colchester, Essex CO4 3SQ, United Kingdom.}
\thanks{Z. Jane Wang is with Department of Electrical and Computer Engineering, The University of British Columbia, Vancouver, BC V6T1Z4, Canada.}
}

\maketitle

\begin{abstract}
This paper considers an intelligent reflecting surface (IRS) assisted multi-input multi-output (MIMO) power splitting (PS) based simultaneous wireless information and power transfer (SWIPT) system with multiple PS receivers (PSRs). 
The objective is to maximize the achievable data rate of the system by jointly optimizing the PS ratios at the PSRs, the active transmit beamforming (ATB) at the access point (AP), and the passive reflective beamforming (PRB) at the IRS, while the constraints on maximum transmission power at the AP, the reflective phase shift of each element at the IRS, the individual minimum harvested energy requirement of each PSR, and the domain of PS ratio of each PSR are all satisfied. 
For this unsolved problem, however, since the optimization variables are intricately coupled and the constraints are conflicting, the formulated problem is non-convex, and cannot be addressed by employing exist approaches directly.  
To this end, we propose a joint optimization framework to solve this problem. Particularly, we reformulate it as an equivalent form by employing the Lagrangian dual transform and the fractional programming transform, and decompose the transformed problem into several sub-problems. Then, we propose an alternate optimization algorithm by capitalizing on the dual sub-gradient method, the successive convex approximation method, and the penalty-based majorization-minimization approach, to solve the sub-problems iteratively, and obtain the optimal solutions in nearly closed-forms.
Numerical simulation results verify the effectiveness of the IRS in SWIPT system and indicate that the proposed algorithm offers a substantial performance gain.
\end{abstract}


%
\IEEEpeerreviewmaketitle

\section{Introduction}
%
%
%
%
Radio frequency (RF) signals can transfer both information and energy simultaneously, which makes it possible to combine wireless power transmit (WPT) and wireless information transmit (WIT) \cite{8476159}. Motivated by this fact, currently, simultaneous wireless information and power transfer (SWIPT) \cite{9145622} has been regarded as an appealing technology for the internet of things (IOT) with low-power and energy-limited devices \cite{6957150}. However, SWIPT also brings new challenges on the trade-off of information decoding (ID) and energy harvesting (EH) operations, which needs to be addressed in the SWIPT system \cite{8214104,8694785,8879665}. In SWIPT system ID receivers and EH receivers can be separated or integrated \cite{funda,8421584}. For the separated architecture, the ID receivers and the EH receivers are different devices. The former is only able to process information and the latter is only for power charging. In contrast, for the integrated architecture, each receiver has integrated circuitry to perform both ID and EH operations, 
by power splitting (PS) \cite{2014Joint} or time switching (TS). In particular, PS receivers (PSRs) split the received signal into two streams of different power, one for ID operation and the other for EH, and TS receivers (TSRs) split the received signal in time domain.

Compared with the conventional receivers, the receivers in SWIPT is more limited by the received energy strength, as a portion of the power must be used to active the circuit. This leads to a performance bottleneck in SWIPT system.
MIMO technology can significantly improve the performance of the SWIPT systems by providing increased link capacity and spectral efficiency \cite{7009979,7862926}. Employing relays is another way to enhance the performance of the SWIPT systems, which can introduce additional links to strengthen the received signal power at the receivers.

Remarkably, there is another promising and cost-effective technique, namely intelligent reflecting surface (IRS), to improve the spectral and energy efficiency of wireless communications systems, due to the fact that IRS can introduce several reflective channels to mitigate detrimental propagation environment. 
In general, an IRS can adjust the wireless propagation environment nearly-instantaneously, by using a vast number of nearly-passive reflective elements, each of which can change the phase of the signals without introducing additional noise \cite{9146875,9140329,di2019smart,9181610}. 
 Optimal PRB design \cite{9117093,9239335,9293155}, which can maximize the achievable data rate, is the major challenge for the IRS-assisted communications systems.
The work \cite{9025235} considered a two-way communications with the IRS and maximized the data rate by employing the Arimoto-Blahut algorithm. 
The authors in \cite{8982186} studied the sum-rate maximization problem under MISO setting and proposed an efficient algorithm based on the vector versions of Lagrangian dual transform (LDT) and fractional programming transform (FPT) \cite{shen2018fractional, shen2018fractional2,8862850}. 
The work \cite{9154244} employed the same vector versions of LDT and FPT as those in \cite{8982186} to solve this sum-rate maximization problem in the more general and complex cell-free scenario. 
In \cite{9076830}, by proposing an majorization-minimization (MM) based \cite{MM} algorithm, the authors maximized the sum-rate of all groups for an IRS-assisted multi-group multi-cast communications system.
Moreover, the authors in \cite{panmulticell} and \cite{9279253} considered the sum-rate maximization problem in multi-cell multi-IRS scenario. 
Particularly, the authors in \cite{panmulticell} employed the well-known weighted minimum mean-square error (WMMSE) \cite{5756489} technique to transform the original sum-rate maximization problem into an equivalent form, and solved it based on the block coordinate descent (BCD) technique, the Lagrangian multiplier method and the Complex Circle Manifold (CCM)/MM algorithms.
The authors in \cite{9279253} proposed an algorithm based on the second-order cone programming (SOCP) and semi-definite relaxation (SDR) \cite{luo2010semidefinite} techniques.
Furthermore, other wide ranges of topics in IRS-assisted communications system had also been studied, such as power allocation, physical layer security, non-orthogonal multiple access, device-to-device communication, and so on \cite{8743496,9316920,9417307,8955968,9348211}. 
Recently, deploying the IRS into the SWIPT system has been attracting increasing attentions. The authors in \cite{zhou2020user} considered deploying IRS to enhance the physical layer security in SWIPT system.
The authors in \cite{xu2020resource} proposed an efficient penalty-based algorithm developed through SCA and SDR approaches to optimize resource allocations in an IRS-assisted SWIPT system, in which a dedicated energy-carrying signal was required, and the EH receiver model was non-linear. 
Meanwhile, the authors in \cite{wsp} proved the dedicated energy signal can only slightly improve the performance, and provided a SDR-based algorithm.
Moreover, the authors in \cite{qos} investigated a transmission power minimization problem under quality-of-services (QOS) constraints, and proposed a penalty-based optimization method to solve it.
The achievable data rate maximization problem with separated ID and EH receivers was investigated in \cite{panswipt}, in which the authors transformed the data rate maximization problem as an equivalent WMMSE minimization problem by exploiting the equivalence between the data rate and WMMSE, and proposed an efficient iterative algorithm based on BCD and MM technologies.
The paper \cite{9257429} studied an IRS-aided SWIPT system with multiple integrated PSRs, and introduced an energy efficiency indicator (EEI) to balance the date rate and harvested energy. Then an efficient algorithm developed by SDR technique, MM algorithm and Dinkelbach approach, was proposed for addressing this problem.
In \cite{9423652}, the authors studied the max-min energy efficiency of the IRS-assisted SWIPT system, and proposed an efficient alternate optimization (AO) algorithm by capitalizing on the penalty-based method, SDR technique and MM approach.

Most exiting works on IRS-assisted MIMO SWIPT systems focus on minimizing the transmission power, or maximizing the received power, or maximizing energy efficiency \cite{zhou2020user,xu2020resource,wsp,qos,9257429,9423652}. For the work on IRS-assisted MIMO SWIPT system focusing on maximizing sum-rate, the receivers for ID and EH are separated devices \cite{panswipt}. 
In this paper, however, we consider the achievable data rate maximization problem in the IRS-assisted MIMO SWIPT system with multiple integrated PSRs, in which the receivers can perform both ID and EH. To our best knowledge, this problem has not been addressed yet, and is fundamentally different from the above mentioned problems in  IRS-assisted SWIPT systems. As we will see later, the optimization problem investigated in this paper is non-convex since the optimization variables are intricately coupled, and the constraints are conflicting, which impose a big challenge for solving the problem.

The main contributions of this paper are summarized as follows
\begin{itemize}
    \item In order to solve this non-convex and challenging optimization problem, we propose a joint optimization framework. Particularly, we employ the matrix versions of LDT and FPT, to transform the problem into an equivalent form, in which the optimization variables are decoupled.
    Given the fact that the objective function of transformed problem with respect to each variable is concave with the others being fixed, consequently, we decompose the transformed problem into several sub-problems. Then, we propose an AO-based algorithm to solve this sub-problems and obtain the solutions of all variables in nearly closed-forms. 
    We also prove that the finally solutions satisfy the Karush-Kuhn-Tucker (KKT) conditions of the original problem;
    
    \item At the PSRs side, we design of the optimal PSRs, i.e., the optimal PS ratios for the EH and ID operations, and the optimal decoding filters (matrices) for decoding the received signals. It is worth to mention here that when optimizing the PS ratio of each PSR for the EH and ID operations, each iteration will yield an upper bound of PS ratios, which can be used to check the feasibility of the problem;
    
    \item 
    At the AP side, first, by approximating minimum harvested energy requirement, which is a non-convex constraint, to a linear form, we reformulate the optimization sub-problem of ATB as a convex form.
    Then we provide a nearly closed-form solution of ATB by exploiting Lagrangian dual sub-gradient method and also prove that the final converged solution of ATB satisfies the KKT conditions of the original sub-problem;
    
    \item At the IRS side, we approximate the constraint of minimum harvested energy by adopting SCA approach and then transform the optimization sub-problem of PRB as a quadratic programming (QP) form under constant modulus constraints after further matrix manipulations. 
    Since the constant modulus constraint of each reflective element is non-convex, the QP problem is still non-convex, and the dual gap is not guaranteed to be zero. 
    To solve this problem with low-complexity, we propose an efficient penalty-based MM algorithm and obtain the solution of PRB in a nearly closed-form. 
    It can also be readily verified that the final converged solution satisfies the KKT conditions of the original sub-problem of optimizing PRB;
    
    \item Finally, simulation results indicate that the IRS can provide a significant performance gain, and the proposed algorithm can substantially enhance the performance. 
\end{itemize}
\textit{Notations:} Boldface low case letters denote vectors and upper case letters stand for matrices. ${\mathbb{C}^{M \times N}}$ represents a complex matrix with the dimensional of ${{M \times N}}$. $\mathbf I _N$ and $\mathbf 0$ denote the $N \times N$ identity matrix and zero matrix.
For two matrices $\mathbf A$ and $\mathbf B$, $\mathbf A \odot \mathbf B$ and $\mathbf A \otimes \mathbf B$ are the Hadamard product and Kronecker product of $\mathbf A$ and $\mathbf B$, respectively. For a square matrix $\mathbf A$, $\mathbf A^{\operatorname{H}}$, $\mathbf A^{\operatorname{T}}$ and $\operatorname{Tr} \left(\mathbf A\right)$ denote the Hermitian conjugate transpose, the transpose, and the trace of matrix $\mathbf A$, respectively.
$\operatorname{diag}\left( \cdot \right)$ denotes the diagonal operation and $\Re\left\{\cdot\right\}$ denotes the real part of a complex number.
Meanwhile, $\mathcal O$ denotes the computational complexity notation.

\section{system model and problem formulation}
\begin{figure}[t]
    \centering
    \includegraphics[width=0.95\linewidth]{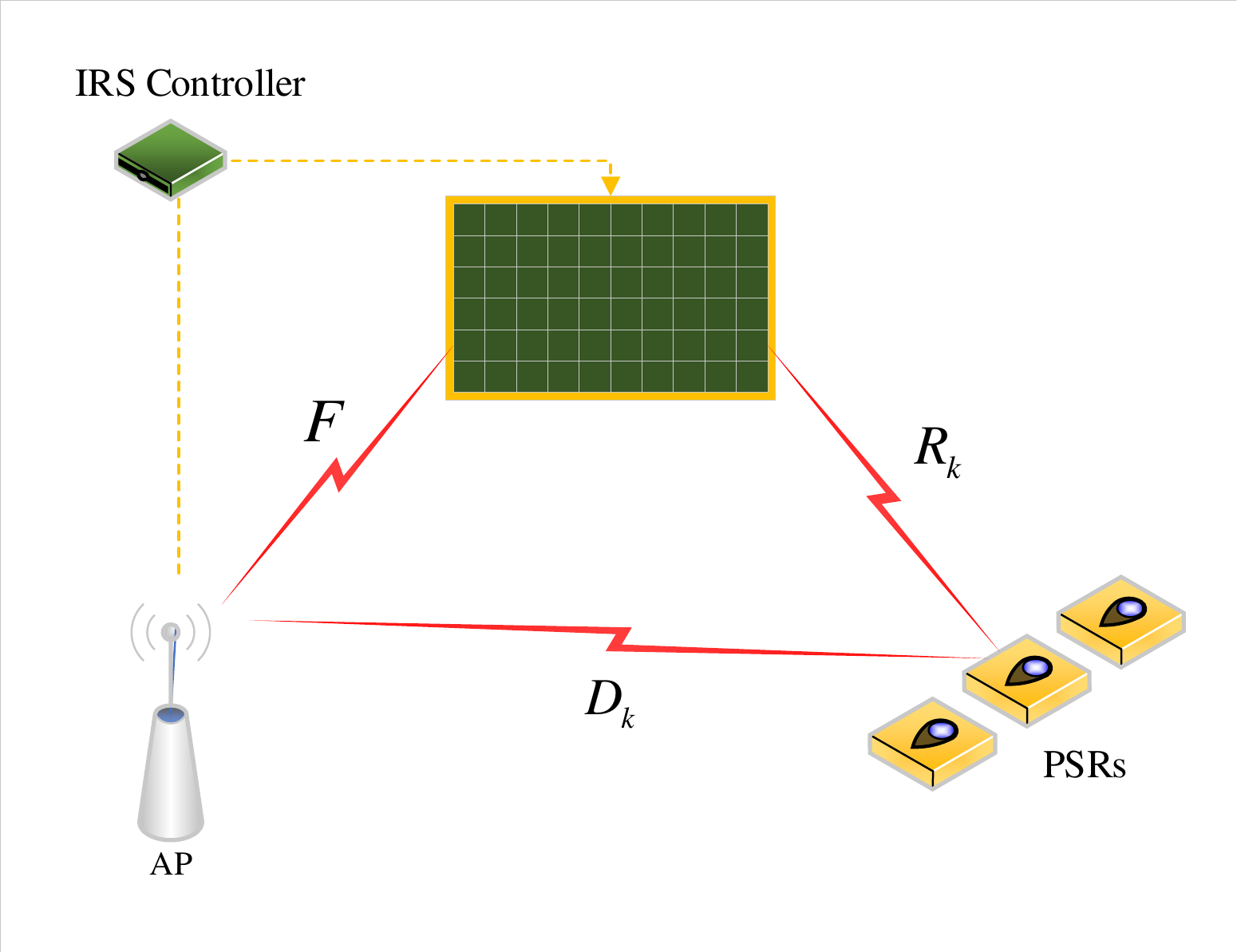}
    \caption{IRS-assisted PS-SWIPT System.}
 \label{f1} 
\end{figure}

In this section, as shown in Fig. \ref{f1}, we describe an IRS-assisted SWIPT system with multiple integrated PSRs. Meanwhile, it is assumed that the AP and the IRS controller can acquire the channel state information (CSI) perfectly \cite{9214532}.\footnote{The conventional channel estimation methods can directly applied for the IRS and there are others channel estimation technologies designed only for IRS system. However, in fact, the required CSI of IRS channels cannot be accurate, therefore, the robust optimization algorithm need to be proposed \cite{9117093}, which is left for future work.}
We define $M_b$ and $M_u$ as the number of the AP antennas and the PSRs antennas, respectively. 
Meanwhile, $N$ denotes the number of reflective elements of the IRS.

The signal at the AP, $\mathbf x$, is a combination of symbols intended to $K$ PSRs, which is given as
$\mathbf x = \sum_{k = 1}^K {{\mathbf W _k}{\mathbf s_k}}$, where $\mathbf s_k \in \mathbb{C}^{M_u \times 1}$ is the data symbol vector, and satisfies $\mathcal C\mathcal N\left(\mathbf 0, \mathbf I_{M_u}\right)$, and $\mathbf W _{k} \in {\mathbb{C}^{M_b \times M_u}}$ denotes the corresponding ATB matrix for the $k^{th}$ PSR.   

The received signal at the $k^{th}$ PSR can be expressed as
\begin{align}
    {\mathbf y_k} = {\mathbf H_k^{\operatorname{H}}}{\mathbf W_k}{\mathbf s_k} + \sum_{i = 1,i \ne k}^K {{\mathbf H_k^{\operatorname{H}}}{\mathbf W_i}{\mathbf s_i}}  + {\mathbf n_k},\forall k,\label{e1}
\end{align}
where $\mathbf n_{k} \sim  \mathcal C \mathcal{N}\left( {\mathbf 0,{\sigma_k ^2\mathbf {I}_{M_u}}} \right)$ denotes the antenna noise vector at the $k^{th}$ PSR, and $\mathbf H_{k} \in \mathbb{C}^{M_b \times M_u}$ denotes the cascade channel from the AP to the $k^{th}$ PSR, and can be mathematically expressed as ${\mathbf H_k^{\operatorname{H}}} = {\mathbf D_k^{\operatorname{H}}} + {\mathbf R_k^{\operatorname{H}}}\Theta \mathbf F$, where ${\mathbf D_k \in {\mathbb{C}^{M_b \times M_u}}}$ denotes the direct channel from the AP to the $k^{th}$ PSR, ${\mathbf R_k \in \mathbb{C}^{N \times M_u}}$ is the channel from the IRS to the $k^{th}$ PSR and the channel from the AP to the IRS is represented by ${\mathbf F \in {\mathbb{C}^{N \times M_b}}}$. Since the signals impinging by the IRS and will be absorbed, which may lead to a power loss, thus, the PRB of IRS can be represented by 
\begin{align}
    \Theta  = \alpha\operatorname{diag}\left( {{e^{j{\phi_1}}},{e^{j{\phi_2}}},\cdots,{e^{j{\phi_N}}}} \right),\label{e2}
\end{align}
where $\phi_n \in [0,2\pi)$, $\forall n \in N$ is the phase shift of $n^{th}$ reflective element, and $\alpha$ denotes reflecting efficiency of IRS.
\footnote{Some previous works assume a more ideal model of IRS as $\Theta =  \operatorname{diag}\left({\varphi_1 {e^{j{\phi_1}}},{\varphi_2e^{j{\phi_2}}},\cdots,{\varphi_Ne^{j{\phi_N}}}} \right)$, where $\varphi_n$ denotes amplitude of $n^{th}$ element of IRS. In this ideal model, IRS can reconfigure the incident signal by changing its amplitude and phase shift. Clearly, the ideal IRS model can offer a higher adaptability to reconfigure the incident signals, however, it may also lead to a higher implementation complexity. The trade-off between performance and implementation complexity is a key problem, which is left for future work.}

\begin{figure}[t]
    \centering
    \includegraphics[width=0.95\linewidth]{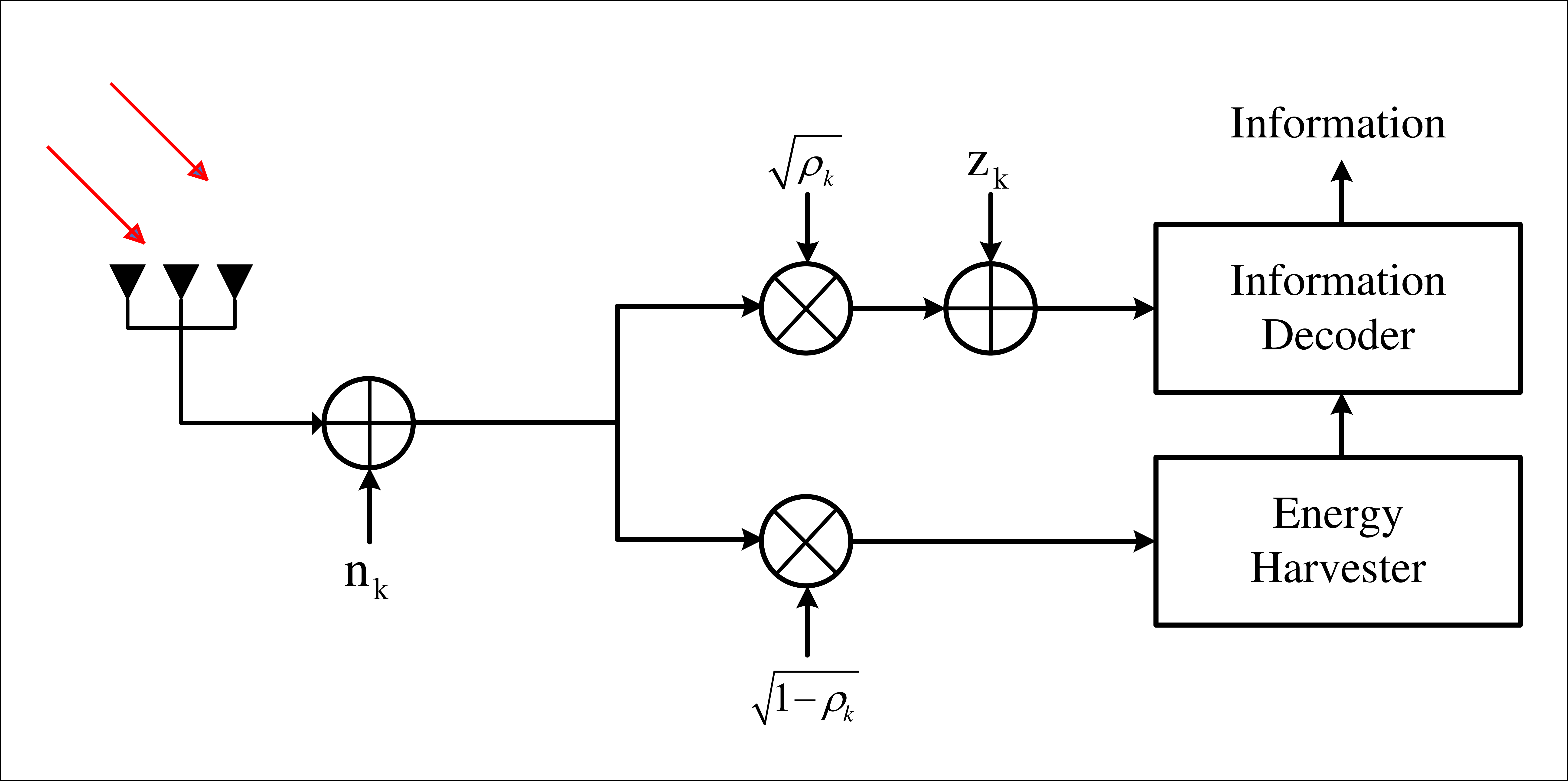}
    \caption{Architecture of Integrated PSR.}
 \label{integrated PSR} 
\end{figure}

In this paper, as shown in Fig. \ref{integrated PSR}, each PSR is an integration of an ID receiver and a EH receiver, and the received signal power is split into two streams by an adjustable PS ratio of $\rho_k\in (0,1)$, which is dedicated for ID operation, and $1-\rho_k$ is exploited for EH operation. Accordingly, the received signal power splits for ID operation can be expressed as
\begin{align}
    \mathbf y_k^{\operatorname{ID}} = \sqrt {{\rho _k}} \left( {\mathbf H_k^{\operatorname{H}}}{\mathbf W_k}{\mathbf s_k} + \sum_{i = 1,i \ne k}^K {{\mathbf H_k^{\operatorname{H}}}{\mathbf W_i}{\mathbf s_i}}  + {\mathbf n_k} \right) + {\mathbf z_k},\forall k,\label{e3}
\end{align}
where $\mathbf z_k\sim  \mathcal C \mathcal{N}\left( {\mathbf 0,{\delta_k ^2\mathbf {I}_{M_u}}} \right)$ is the additional noise generated during signal processing. Therefore, the signal to interference plus noise ratio (SINR) at the $k^{th}$ PSR can be formulated as
\begin{align}
    {\Gamma _k} = \frac{{{\rho _k}{\mathbf H_k^{\operatorname{H}}}{\mathbf W_k}\mathbf W_k^{\operatorname{H}}\mathbf H_k}}{{{\rho _k}\sum_{i = 1,i \ne k}^K {{\mathbf H_k^{\operatorname{H}}}{\mathbf W_i}\mathbf W_i^{\operatorname{H}}\mathbf H_k}  + {\rho _k}\sigma _k^2\mathbf I_{M_u} + \delta _k^2\mathbf I_{M_u}}},\forall k.\label{e4}
\end{align}
The received signal power splits for EH operation is given as 
\begin{align}
    \mathbf y_k^{\operatorname{EH}} = \sqrt {1 - {\rho _k}} \left( {\sum_{i = 1}^K{{\mathbf H_k^{\operatorname{H}}}{\mathbf W_i}{\mathbf s_i}}  + {\mathbf n_k}} \right),\forall k.\label{e5}
\end{align}
By ignoring the antenna noise power \cite{9423652}, i.e., $\sigma_k^2$, since it is negligible, the split power for EH can be expressed as
\begin{align}
    {\mathbf e_k} = {\eta _k}\left( {1 - {\rho _k}} \right)\operatorname{Tr}\left( \sum_{i = 1}^K {\mathbf H_k^{\operatorname{H}}}{\mathbf W_i}\mathbf W_i^{\operatorname{H}}\mathbf H_k \right),\forall k,\label{e6}
\end{align}
where $\eta _k \in \left(0, 1\right)$ is the energy conversion efficiency for the $k^{th}$ PSR. 

The achievable data rate of all the PSRs can be expressed as $\mathcal R\left( {\rho, \mathbf W,\Theta } \right) = \sum_{k = 1}^K {\log \left| {\mathbf I_{M_u} + {\Gamma _k}} \right|}$.
Accordingly, the achievable data rate maximization problem can be defined as
\begin{align}
\mathop {\max }_{\rho, \mathbf W,\Theta } \quad & \mathcal R\left( {\rho, \mathbf W, \Theta } \right) \label{e7}\\
\operatorname{s.t.}\quad &\sum_{k = 1}^K {\left\| {{\mathbf W_k}} \right\|_2^F \le {P_{\max }}}, \tag{7a} \label{4a}\\
\quad &\operatorname{Tr}\left( {\sum_{i = 1}^K {{\mathbf H_k^{\operatorname{H}}}{\mathbf W_i}\mathbf W_i^{\operatorname{H}} \mathbf H_k} } \right) \ge \frac{\mathbf {\bar e}_k}{{\eta _k}\left( {1 - {\rho _k}} \right)}, \forall k,\tag{7b} \label{4b}\\
\quad &\left| \Theta _{n, n} \right|  = \frac{1}{\alpha}, \forall n, \tag{7c} \label{4c}\\
\quad & 0 < {\rho _k} < 1, \forall k,\tag{7d} \label{4d}
\end{align}
where the constraint \eqref{4a} limits the maximum transmission power $P_{\max}$ of the AP, the constraints \eqref{4b} are imposed to guarantee the harvested power is larger than minimum EH requirement $\mathbf {\bar e}_k, \forall k$, and the constraint \eqref{4d} is the bound of the PS ratio of each PSR, and \eqref{4c} represents the constant modulus constraints of reflective elements at the IRS.

It is clear that Problem \ref{e7} is non-convex and intractable, due to the fact that all the variables are intricately coupled in a function form of logarithmic, and the minimum EH constraint of each PSR is also non-convex and conflicting with maximum transmission power constraint at the AP. 
Additionally, the constraints on the PS ratios impose another challenge for the optimization problem. To this end, in the next section, we provide an efficient joint optimization framework to solve the problem.

\section{Proposed Joint Optimization Framework}

First, we move the matrix-ratio terms out of the logarithm in the objective function of Problem \ref{e7} by employing the matrix version of LDT, and we have the following proposition

\begin{proposition}[Matrix version of LDT]
By introducing an auxiliary diagonal matrix $\mathbf U_k\in \mathbb{C}^{M_u \times M_u}$, $\forall k$, Problem \ref{e7} can be reformulated as an equivalent form
\begin{align} 
    \mathop {\max }_{\rho ,\mathbf W,\Theta ,\mathbf U} \quad & {f_1}\left( {\rho ,\mathbf W,\Theta ,\mathbf U} \right) \label{eq8}\\
    \operatorname{s.t.} \quad 
    & \eqref{4a}-\eqref{4d}, \notag 
\end{align}
where the new objective function can be written as
\begin{align}
    {f_1}\left( {\rho, \mathbf W,\Theta, \mathbf U} \right) &= \sum_{k=1}^K {\left( \log \left| {\mathbf I_{M_u} + {\mathbf U_k}} \right|  - \operatorname{Tr}\left( \mathbf U_k \right)\right)} \notag\\
    &+\sum_{k = 1}^K {\operatorname{Tr}\left(\left( {\mathbf I_{M_u} + {\mathbf U_k}} \right){f_2}\left( {\rho, \mathbf W,\Theta } \right)\right)},
\end{align}
where ${f_2}\left( {\rho,\mathbf W,\Theta} \right)$ is given by
\begin{align}
    {f_2}\left( {\rho,\mathbf W,\Theta} \right) = {\mathbf H_k^{\operatorname{H}}}{\mathbf W_k}\mathbf V_k^{ - 1}\mathbf W_k^{\operatorname{H}}\mathbf H_k,
\end{align}
with ${\mathbf V_k} = \sum_{i = 1}^K {{\mathbf H_k^{\operatorname{H}}}{\mathbf W_i}\mathbf W_i^{\operatorname{H}}\mathbf H_k}  + \sigma _k^2\mathbf I_{M_u} + {\delta _k^2}/{\rho _k}\mathbf I_{M_u}$, $\forall k$. 
\end{proposition}
\begin{proof}
The proposition is the matrix version of LDT, and the detailed proof of this reformulation method is similar to the vector version \cite{shen2018fractional,shen2018fractional2,8862850}.
\end{proof}
Although, the matrix-ratio term, i.e., SINR, has been moved out of the logarithm function, and the transformed problem has been considerably simplified, the transformed problem is still non-convex and difficult to be handled directly since the optimization variables are still coupled in the matrix-ratio term. Hence, we reformulate the transformed problem as the linear form as follows.


\begin{proposition}[Matrix version of FPT]\label{fpt}
By introducing auxiliary diagonal matrix $\mathbf L_k\in \mathbb{C}^{M_u \times M_u}$, $\forall k$, the above sub-problem can be equivalently transformed to
\begin{align}
    \mathop {\max }_{\mathbf W, \mathbf L} \quad & {f_3}\left( {\rho, \mathbf W, \Theta, \mathbf U, \mathbf L} \right) \label{e20} \\
    \operatorname{s.t.} \quad 
    & \eqref{4a}-\eqref{4d}, \notag
\end{align}
where ${f_3}\left( {\rho, \mathbf W, \Theta, \mathbf U, \mathbf L} \right)$ can be expressed as 
\begin{align}\label{ep12}
&{f_3}\left( {\rho,\mathbf W,\Theta ,\mathbf U,\mathbf L } \right)  = {\sum _{k = 1}^K {\log \left| { {\mathbf {\bar U} _{k}}} \right|} }  - {\sum _{k = 1}^K { \operatorname{Tr} \left( {{\mathbf U _{k}}} \right)} }\notag\\
 &+{\sum _{k = 1}^K {\operatorname{Tr} \left( { {{\mathbf {\bar U} _{k}}}\mathbf H_k^{\operatorname{H}}\mathbf W_k \mathbf L_k^{\operatorname{H}}}\right)}}+{\sum _{k = 1}^K {\operatorname{Tr} \left( { {{\mathbf {\bar U} _{k}}}\mathbf L_k\mathbf W_k^{\operatorname{H}}\mathbf H_k}\right)}} \notag\\
 &-{\sum _{k = 1}^K {\operatorname{Tr} \left( { {{\mathbf {\bar U} _{k}}}\mathbf L_k\sum_{i=1}^K{\mathbf H_k^{\operatorname{H}}\mathbf W_i\mathbf W_i^{\operatorname{H}}\mathbf H_k}\mathbf L_k^{\operatorname{H}}}\right)}}\notag\\
 &-{\sum _{k = 1}^K {\operatorname{Tr} \left( { {{\mathbf {\bar U} _{k}}}\mathbf L_k\left(\sigma _k^2\mathbf I_{M_u} + {\delta _k^2}/{\rho _k}\mathbf I_{M_u}\right)\mathbf L_k^{\operatorname{H}}}\right)}},
\end{align}
where ${\mathbf {\bar U} _{k}}\triangleq{\mathbf {U} _{k}}+\mathbf I_{M_u},\forall k$.
\end{proposition}
\begin{proof}
The proposition is extend FPT from the vector into the matrix version, the details proof please refer to  \cite{shen2018fractional,shen2018fractional2,8862850}.
\end{proof}

Note that the optimization variables in Problem \ref{e20} are decoupled, and the objective function in \eqref{ep12} of Problem \ref{e20} is concave with respect to any one of $\rho$, $\mathbf W$, $\Theta$, $\mathbf U$ and $\mathbf L$, with the others being fixed. Based on this fact, we decompose Problem \ref{e20} into several sub-problems, and solve them alternately at three sides, i.e., the PSRs side, the AP side and the IRS side. In the $t$-th iteration, we have

\subsection{Optimization at the PSRs side}
Clearly, to design the optimal PSRs, we need to optimize three variables related to the PSRs, i.e., $\mathbf U$, $\rho$ and $\mathbf L$.

First, we consider to optimize the auxiliary diagonal matrix of $\mathbf U_k,\forall k$. With fixed $\rho^{\left(\operatorname{t}\right)}$, $\mathbf W ^{\left(\operatorname{t}\right)}$, $\Theta ^{\left(\operatorname{t}\right)}$ and $\mathbf L^{\left(\operatorname{t}\right)}$, the solution of $\mathbf U$ can be obtained by solving the following sub-problem
\begin{align}\label{eq13}
    \mathbf U^{\left(\operatorname{t+1}\right)} \triangleq \arg\mathop{\max} _\mathbf U {f_3}\left( \rho^{\left(\operatorname{t}\right)}, {\mathbf W^{\left(\operatorname{t}\right)},\Theta^{\left(\operatorname{t}\right)} ,\mathbf U,\mathbf L^{\left(\operatorname{t}\right)} } \right).
\end{align}
Note that, the variable $\mathbf U$ only appears in the objective function of the Sub-problem \ref{eq13} and does not exist in any constraint set. Therefore, by setting the partial derivatives of the objective function in Sub-problem \ref{eq13} with respect to $\mathbf U_k,\forall k$ to be zero, and after some matrix manipulations, the closed-form solution is given as
\begin{align}\label{eq14}
    \mathbf U_k^{\left(\operatorname{t+1}\right)}= \Gamma _k, \forall k.
\end{align}
Note that the variables $\rho$, $\mathbf W$, $\Theta$ and $\mathbf L$ only exists in some terms of ${f_3}\left( {\mathbf W,\Theta ,\mathbf U,\mathbf L} \right)$,
with the obtained $\mathbf U^{\left(\operatorname{t+1}\right)}$ in \eqref{eq14}, we can recast the objective function as $f_3\left( {\rho,\mathbf W,\Theta ,\mathbf U,\mathbf L} \right)\triangleq f_4\left( {\rho,\mathbf W,\Theta ,\mathbf L} \right)+\operatorname{Const}\left(\mathbf U\right)$, where 
\begin{align}
    &f_4\left( {\rho,\mathbf W,\Theta,\mathbf L} \right)=\sum _{k = 1}^K {\operatorname{Tr} \left( { {{\mathbf {\bar U} _{k}}}\mathbf H_k^{\operatorname{H}}\mathbf W_k \mathbf L_k^{\operatorname{H}}}\right)}\notag\\
    &+\sum_{k=1}^K{\operatorname{Tr} \left( { {{\mathbf {\bar U} _{k}}}\mathbf L_k \mathbf W_k^{\operatorname{H}}\mathbf H_k}\right)}\notag\\
    &-{\sum _{k = 1}^K {\operatorname{Tr} \left( { {{\mathbf {\bar U} _{k}}}\mathbf L_k\sum_{i=1}^K{\mathbf H_k^{\operatorname{H}}\mathbf W_i\mathbf W_i^{\operatorname{H}}\mathbf H_k}\mathbf L_k^{\operatorname{H}}}\right)}}\notag\\
    &-{\sum _{k = 1}^K {\operatorname{Tr} \left( { {{\mathbf {\bar U} _{k}}}\mathbf L_k\left(\sigma _k^2\mathbf I_{M_u} + {\delta _k^2}/{\rho _k}\mathbf I_{M_u}\right)\mathbf L_k^{\operatorname{H}}}\right)}},
\end{align} 
and $\operatorname{Const}\left(\mathbf U\right)=\sum_{k=1}^K{\log\left|\mathbf {\bar U}_k\right|-\operatorname{Tr}\left(\mathbf U_k\right)}$.
Therefore by fixing the obtained variable of $\mathbf U_k, \forall k$ in \eqref{eq14}, we can optimize other variables by only investigating $f_4$. 

Particularly, we can design the optimal PS ratios of PSRs and check the feasibility of original problem. With fixed $\mathbf W^{\left(\operatorname{t}\right)}$, $\Theta^{\left(\operatorname{t}\right)}$ and $\mathbf L^{\left(\operatorname{t}\right)}$, and the obtained $\mathbf U^{\left(\operatorname{t+1}\right)}$ in \eqref{eq14}, the sub-problem for optimizing $\rho$ is equivalent to 
\begin{align}\label{eq16}
    \rho^{\left(\operatorname{t+1}\right)} \triangleq \arg\mathop{\max} _\rho \quad&{f_4}\left( {\rho,\mathbf W^{\left(\operatorname{t}\right)},\Theta^{\left(\operatorname{t}\right)},\mathbf L^{\left(\operatorname{t}\right)} } \right)\\
    \operatorname{s.t.}\quad& \eqref{4b},\; \eqref{4d},\notag
\end{align}
and we have the following lemma
\begin{lemma}\label{l1}
The objective function of Problem \ref{eq16} is concave and monotonous increasing with respect to $\rho$.
\end{lemma}
\begin{proof}
The proof is presented in Appendix \ref{app1}.
\end{proof}

According to Lemma \ref{l1}, and considering \eqref{4b}, \eqref{4d} and the objective function in Problem \ref{eq16} together, $\rho_k,\;\forall k$ is required to satisfy the following condition \cite{9104736,2017Joint}
\begin{align}
   \rho _k^{\left(t+1\right)} = 1 - \frac{{{\mathbf {\bar e}_k}}}{{{\eta _k}\operatorname{Tr}\left( {\sum_{i = 1}^K {{\mathbf H_k^{\operatorname{H}}}{\mathbf W_i}\mathbf W_i^{\operatorname{H}}\mathbf H_k} } \right)}}  , \forall k,\label{eq17}
\end{align}
and we have the following lemma
\begin{lemma}\label{l2}
Problem \ref{eq16} is feasible when $\rho _k^{\left(t+1\right)} > 0, \forall k$.
\end{lemma} 
\begin{proof}
The proof is presented in Appendix \ref{app2}.
\end{proof}

Now, we optimize the variable $\mathbf L_k,\forall k$ with the fixed variables $\mathbf W^{\left(\operatorname{t}\right)}$ and $\Theta^{\left(\operatorname{t}\right)}$, and the obtained solutions of $\mathbf U^{\left(\operatorname{t+1}\right)}$ and $\rho^{\left(\operatorname{t+1}\right)}$, by using \eqref{eq14} and \eqref{eq17}, respectively. The sub-problem of optimization for $\mathbf L$ is equivalent to 
\begin{align}
    \mathbf L^{\left(\operatorname{t+1}\right)} \triangleq \arg\mathop{\max} _\mathbf L {f_4}\left( {\rho^{\left(\operatorname{t+1}\right)}, \mathbf W^{\left(\operatorname{t}\right)},\Theta^{\left(\operatorname{t}\right)},\mathbf L } \right).
\end{align}
Similar to the method for solving the sub-problem of $\mathbf U$ in \eqref{eq13}, the variable $\mathbf L$ does not exist in the constraint set, hence, by setting the partial derivatives of the objective function of the above sub-problem to be zero, the optimal solution of $\mathbf L$ can be obtained as
\begin{align}\label{eq18}
    \mathbf L_k^{\left(\operatorname{t+1}\right)}=\mathbf V_k^{-1}\mathbf H_k^{\operatorname{H}}\mathbf W_k, \forall k.
\end{align}

Note that $\mathbf L_k,\forall k$ can be treated as the decoding matrix to decode the received signals at the PSRs side, which leads to the minimum mean-square error (MMSE). For details information, please refer to Appendix \ref{app3}. 

Aforementioned, we complete the optimal design at the PSRs side.

\subsection{Optimization at the AP side}
At the AP side, with fixed $\Theta^{\left(\operatorname{t}\right)}$, and the obtained solutions of variables $\mathbf U^{\left(\operatorname{t+1}\right)}$, $\rho^{\left(\operatorname{t+1}\right)}$ and $\mathbf L^{\left(\operatorname{t+1}\right)}$, the optimization sub-problem of ATB is defined as  
\begin{align}
    \mathbf W^{\left(\operatorname{t+1}\right)} \triangleq \arg\mathop{\max} _\mathbf W& \; {f_4}\left( {\rho^{\left(\operatorname{t+1}\right)}, \mathbf W,\Theta^{\left(\operatorname{t}\right)} , \mathbf L^{\left(\operatorname{t+1}\right)} } \right)\\
    \operatorname{s.t.}& \; \eqref{4a},\;\eqref{4b}.\notag
\end{align}
By omitting the irrelevant constant term, i.e., ${\sum _{k = 1}^K {\operatorname{Tr} \left( { {{\mathbf {\bar U} _{k}}}\mathbf L_k\left(\sigma _k^2\mathbf I_{M_u} + {\delta _k^2}/{\rho _k}\mathbf I_{M_u}\right)\mathbf L_k^{\operatorname{H}}}\right)}}$, which has no impact on the update of $\mathbf W$, the above optimization problem can be simplified  as
\begin{align}\label{eq21}
    \mathop{\min} _\mathbf W &\; {\sum _{k = 1}^K {\operatorname{Tr} \left( { {{\mathbf {\bar U} _{k}}}\mathbf L_k\mathbf H_{k}^{\operatorname{H}}{\sum _{j = 1}^K {{\mathbf W_{j}}\mathbf W_{j}^{\operatorname{H}}} }{\mathbf H_{k}}\mathbf L_k^{\operatorname{H}}}\right)}}\notag\\
    &-{\sum _{k = 1}^K {\operatorname{Tr} \left( { {{\mathbf {\bar U} _{k}}}\mathbf H_k^{\operatorname{H}}\mathbf W_k \mathbf L_k^{\operatorname{H}}}\right)}}-{\sum _{k = 1}^K {\operatorname{Tr} \left( { {{\mathbf {\bar U} _{k}}}\mathbf L_k \mathbf W_k^{\operatorname{H}}\mathbf H_k}\right)}} \\
    \operatorname{s.t.}& \; \eqref{4a},\;\eqref{4b}.\notag
\end{align}
Since the minimum EH constraint \eqref{4b} is non-convex, the sub-problem to obtain $\mathbf W$ is still challenging. We can exploit SCA approach to approximate the constraint \eqref{4b}. 
First, we rewrite the constraint \eqref{4b} as 
\begin{align}
    \operatorname{Tr}\left (\sum_{i=1}^K{\mathbf {W}_i^{\operatorname{H}}\mathbf { B}_k{\mathbf {W}_i}} \right) \ge \mathbf {\bar e}_k/{{\eta_k \left( {1 - {\rho _k}} \right)}} \triangleq {\mathbf {\dot e}_k}, \;\forall k,
\end{align}
where ${\mathbf B_k} \triangleq  {\mathbf H_k}\mathbf H_k^{\operatorname{H}}$. Then, by employing the first-order Taylor expansion, in $n$-th sub-iteration for updating $\mathbf W$, we arrive at 
\begin{align}
    \operatorname{Tr} \left (\sum_{i=1}^K{\mathbf {W}_i^{\operatorname{H}}
    {\mathbf {B}_k}{\mathbf {W}_i}}\right ) \ge  -\operatorname{Tr} \left ( \sum_{i=1}^K{{\left(\mathbf {W}_i^{\left( {n} \right)}\right)}^{\operatorname{H}}{\mathbf{B}_k}{\mathbf {W}_i^{\left( {n} \right)}} }\right ) \quad\notag\\
    + 2 \Re \left\{\operatorname{Tr}\left (\sum_{i=1}^K{{\left(\mathbf {W}_i^{\left( {n} \right)}\right)}^{\operatorname{H}}{\mathbf { B}_k}\mathbf { W}_i} \right )\right\},\forall k. \label{e26}
\end{align}
where $\mathbf W _k^{\left(n\right)}, \forall k$, denotes the solution obtained from the previous sub-iteration. Hence, we can approximate the constraint \eqref{4b} by 
\begin{align}
    2 \Re \left\{\operatorname{Tr}\left (\sum_{i=1}^K{
    {\left(\mathbf {W}_i^{\left( {n} \right)}\right)}^{\operatorname{H}}
    {\mathbf {B}_k}\mathbf {W}_i }\right )\right\} \qquad\qquad\qquad\qquad\notag\\
    \ge {\mathbf {\dot e}_k} + \operatorname{Tr} \left (\sum_{i=1}^K{ {\left(\mathbf {W}_i^{\left( {n} \right)}\right)}^{\operatorname{H}}
    {\mathbf {B}_k}{\mathbf {W}_i^{\left( {n } \right)}} }\right )
    \triangleq {\mathbf {\ddot e} _k^{\left(n\right)}},\forall k.\label{e27}
\end{align}
Similarly, by defining $\mathbf A_k \triangleq  {{\mathbf H_k}{\mathbf {\bar U}_k}{\mathbf L_k}\mathbf L_k^{\operatorname{H}}\mathbf H_k^{\operatorname{H}}}$ and $\mathbf A\triangleq \sum_{k=1}^K{\mathbf A_k}$, and ${\mathbf {S}_k} \triangleq  {\mathbf H_k}{\mathbf {\bar U}_k}{\mathbf L_k}$, we obtain the following problem, which is equivalent to Problem \ref{eq21}
\begin{align}\label{eq25}
    \mathop {\min }_\mathbf {W}\quad & f_5\left(\mathbf W\right)\triangleq\sum_{k=1}^K{\operatorname{Tr}\left(\mathbf {W}_k^{\operatorname{H}}\mathbf {A}
    \mathbf {{W}}_k\right)} - 2\sum_{k=1}^K{\operatorname{Tr}\left(
    \mathbf{W}_k^{\operatorname{H}}\mathbf {S}_k\right)} \\
    \operatorname{s.t.} \quad & \eqref{4a},\;\eqref{e27}. \notag 
\end{align}
Since the objective function $f_5\left(\mathbf W\right)$ and the constraint \eqref{4a} are both convex, and the constraint \eqref{e27} is linear, the problem constitutes a convex optimization problem, which can be solved by standard CVX tools \cite{cvx} directly. However, CVX is employed in each sub-iteration, and this leads to a high computational complexity. Consequently, since the dual gap is guaranteed to be zero, we give the solution of ATB in a nearly closed-form by exploiting Lagrangian dual sub-gradient method. The corresponding Lagrangian function of Problem \ref{eq25} can be written as
\begin{align}
    &\mathcal{L}\left(\mathbf {W}| \mathbf W^{\left(n\right)}, \tau, \mu\right)\triangleq\sum_{k=1}^K{\operatorname{Tr}\left(\mathbf {W}_k^{\operatorname{H}}\mathbf {A}
    \mathbf {{ W}}_k\right)} - 2\sum_{k=1}^K{\operatorname{Tr}\left(
    \mathbf{W}_k^{\operatorname{H}}\mathbf {S}_k\right)}\notag\\
    &+\tau\left(\sum_{k=1}^K{\operatorname{Tr} \left(\mathbf {W}_k^{\operatorname{H}} \mathbf {W}_k \right)}-P_{\max}\right)\notag\\
    &-\sum_{k=1}^K{\mu_k\left(2\Re \left\{\sum_{i=1}^K{\operatorname{Tr}\left (
    {\left(\mathbf {W}_i^{\left( {n} \right)}\right)}^{\operatorname{H}}
    {\mathbf {B}_k}\mathbf {W}_i \right )}\right\}-{\mathbf {\ddot e} _k^{\left(n\right)}}\right)},\label{e29}
\end{align}
where $\tau \ge 0$ and $\mu_k \ge 0, \;\forall k$ are the Lagrangian multipliers corresponding to the constraints \eqref{4a} and \eqref{e27}, respectively. The KKT conditions are given as following
\begin{align}
    \operatorname{Tr}\left(\left({\mathbf {A} +  \tau \mathbf {I}_{M_b}}\right) \mathbf {W}_k-
    \mathbf {S}_k-{\mu_k\mathbf {B}_k}{\mathbf {W}_k^{\left(n\right)}}\right)&= 0,\forall k; \\
    \tau\left( \sum_{k=1}^K{\operatorname{Tr}\left(\mathbf {W}_k^{\operatorname{H}}  \mathbf {W}_k \right)}-P_{\max}\right)&= 0;\\
    \mu_k\left(2\Re \left\{\sum_{i=1}^K{\operatorname{Tr}\left ({\left(\mathbf {W}_i^{\left( {n} \right)}\right)}^{\operatorname{H}}{\mathbf {B}_k}\mathbf {W}_i \right )}\right\}-{\mathbf {\ddot e} _k}\right)&=0,\forall k.\label{e27c}
\end{align}
With fixed Lagrangian multipliers of $\mu _k^{\left(n\right)}$ and $\tau^{\left(n\right)}$ in $n$-th sub-iteration, the optimal solution of ATB can be explicitly expressed as
\begin{align}
    \mathbf {W}_k^{\left(n+1\right)}=\frac{\mathbf {S}_k+{\mu_k^{\left(n\right)}
    \mathbf {B}_k{\mathbf {W}_k^{\left(n\right)}}}}{{\mathbf {A} +  \tau^{\left(n\right)} \mathbf I_{M_b}}},\forall k \in K.\label{e31}
\end{align}
Then, we update Lagrangian multipliers with the optimized $\mathbf {W}_k^{\left(n+1\right)}, \forall k$ by employing \eqref{e31}. Particularly, with fixed $\tau^{\left(n\right)}$, and by defining $\mathbf {\bar A}=\left(\mathbf A +\tau^{\left(n\right)}\mathbf I_{M_b}\right)$, the optimal solution of $\mu _k$ in a closed-form can be determined as
\footnote{
 To obtain the optimal value of $\mu$, we first need to consider the case of $\mu _k = 0, \forall k$ \cite{0Complex}. In this case, if $2\Re \left\{\sum_{i=1}^K{\operatorname{Tr}\left ({\left(\mathbf {W}_i^{\left( {n} \right)}\right)}^{\operatorname{H}}{\mathbf {B}_k}\mathbf {W}_i^{\left(n+1\right)} \right )}\right\}\ge{\mathbf {\ddot e} _k^{\left(n\right)}}$ holds, where $\mathbf {W}_k^{\left(n+1\right)}=\left({{\mathbf {A} +  \tau^{\left(n\right)} \mathbf I_{M_b}}}\right)^{\operatorname{-1}}{\mathbf {S}_k}$, the optimal value of $\mu _k^{\left(n+1\right)}$ is zero. Otherwise, the optimal value of $\mu _k > 0$ can be sequentially determined by employing \eqref{eqmu}.
}
\begin{align}\label{eqmu}
    \mu _k^{\left(n+1\right)}=\frac{{\mathbf {\ddot e} _k}^{\left(n\right)}-2\Re\left\{\sum_{i=1}^K{\operatorname{Tr}\left(\left(\mathbf W_i^{\left(n\right)}\right)^{\operatorname{H}}\mathbf B_k \mathbf {\bar A}^{\operatorname{-1}}\mathbf S_k \right)}\right\}}{2\sum_{i=1}^K{\operatorname{Tr}\left(\left(\mathbf W_i^{\left(n\right)}\right)^{\operatorname{H}}\mathbf B_k\mathbf {\bar A}^{\operatorname{-1}}\mathbf B_k^{\operatorname{H}} \mathbf W_i^{\left(n\right)} \right)}},\forall k.
\end{align}

Next, we propose an efficient sub-gradient based method to obtain the multiplier associated to the constraint \eqref{4a}. Particularly, the multiplier $\tau^{\left(n+1\right)}$ can be updated as following
\begin{align}
    \tau^{\left(n+1\right)}&=\left[\tau^{\left(n\right)}+\upsilon\left(\sum_{k=1}^K{\operatorname{Tr}\left(\mathbf {W}_k^{\operatorname{H}}  \mathbf {W}_k \right)}-P_{\max}\right)\right]^{\operatorname{+}},
\end{align}
where $\upsilon$ represents the step size for updating $\tau$ and $\left[ a \right]^{\operatorname{+}}=\max\left(0,a\right)$ guaranties $\tau^{\left(n+1\right)}$ being positive.


The SCA-based algorithm for optimization at the AP side is given in Algorithm \ref{a1}, the finally solution of ATB by employing Algorithm \ref{a1} can be expressed as $\mathbf W_k^{\left(t+1\right)}\triangleq\mathbf W_k^{\left(\hat n\right)},\forall k$,
where $\hat n$ denotes the maximum number of sub-iterations when Algorithm \ref{a1} converges. Now we have the following Lemma
\begin{lemma}\label{l3}
The sequence of feasible solutions $\left\{\mathbf W^{\left(n\right)}, n=1,2,\cdots, \hat n\right\}$ generated by Algorithm \ref{a1} are optimal solutions of Problem \ref{eq25}. Meanwhile, the final converged solution of Algorithm \ref{a1}, i.e., $\mathbf W^{\left(\hat n\right)}$, also satisfies the KKT conditions of Problem \ref{eq21}.
\end{lemma}

\begin{proof}
The proof is given in Appendix \ref{app4}.
\end{proof}

\begin{algorithm}[t]
\caption{SCA-based Algorithm for Optimizing ATB} 
\label{a1} 
\begin{algorithmic}[1]
\REQUIRE $\mathbf W_{k} ^{\left ( t \right )},\;\forall k$, threshold $\varepsilon$.
\ENSURE     Optimal value of $ \mathbf W_{k}^{\left(t+1\right)},\;\forall k$\\
\STATE\textbf{Set} $\mathbf W_{k} ^{\left ( 0 \right )}\triangleq\mathbf W_{k} ^{\left ( t \right )},\;\forall k$;\\
\STATE\textbf{Calculate} \\$f_5\left(\mathbf W^{\left(0\right)}\right)=\sum_{k=1}^K{\operatorname{Tr}\left({\left(\mathbf {W}_i^{\left( {0} \right)}\right)}^{\operatorname{H}}\mathbf {A}\mathbf  W_k^{\left(0\right)}- 2{\left(\mathbf {W}_i^{\left( {0} \right)}\right)}^{\operatorname{H}}\mathbf {S}_k\right)}$;\\
\FOR{$n=1 \;\text{to}\; 2,3,\cdots,n_{\max}$}
\STATE \textbf{Update} ${\mathbf {\ddot e} _k^{\left(n\right)}},\forall k$, by using \eqref{e27};
\STATE \textbf{Update} $\mathbf W_{k} ^{\left ( n+1 \right )}, \forall\, k$, by using \eqref{e31};
\IF {${\left|f\left(\mathbf W^{\left(n+1\right)}\right) - f\left(\mathbf W^{\left(n\right)}\right) \right|}/{f\left(\mathbf W^{\left(n\right)}\right)} < \varepsilon_1 $} 
\item {Break\;;}
\ENDIF
\ENDFOR
\STATE {$\mathbf{W}_k^{\left(t+1\right)}\triangleq \mathbf W_k^{\left(\hat n\right)},\;\forall k$}.
\end{algorithmic} 
\end{algorithm}
\subsection{Optimization at the IRS side}
Finally, we consider the optimization at the IRS side. By fixing the obtained solutions of $\mathbf U ^{{t+1}}$, $\rho^{{t+1}}$,  $\mathbf L^{{t+1}}$ and $\mathbf W^{{t+1}}$, the corresponding sub-problem for optimizing the PRB of IRS is given as
\begin{align}\label{eq30}
    \Theta^{\left(\operatorname{t+1}\right)} \triangleq \arg\mathop{\max} _\Theta \quad&{f_4}\left( \rho^{\left(\operatorname{t+1}\right)},{\mathbf W^{\left(\operatorname{t+1}\right)},\Theta,\mathbf L^{\left(\operatorname{t+1}\right)} } \right)\\
    \operatorname{s.t.} \quad&\eqref{4b}-\eqref{4c}.\notag
\end{align}
By substituting ${\mathbf H_k^{\operatorname{H}}}= {\mathbf D_k^{\operatorname{H}}} + {\mathbf R_k^{\operatorname{H}}}\Theta \mathbf F$ into $\mathbf H_{k}^{\operatorname{H}}{\mathbf W_{k}}$ and $\mathbf H_{k}^{\operatorname{H}}{\mathbf W_{i}}\mathbf W_{i}^{\operatorname{H}}{\mathbf H_{k}}$, and defining $\mathbf {\hat W}\triangleq\sum_{j=1}^K{\mathbf W_j\mathbf W_j^{\operatorname{H}}}$, we have 
\begin{align}
    {\mathbf H_k^{\operatorname{H}}}{\mathbf W_k} &= {\mathbf D_k^{\operatorname{H}}}{\mathbf W_k} + {\mathbf R_k^{\operatorname{H}}}\Theta \mathbf F{\mathbf W_k},\notag\\
    {\mathbf H_k^{\operatorname{H}}}\mathbf {\hat W}\mathbf H_k &= \mathbf D_k^{\operatorname{H}}\mathbf {\hat W}\mathbf D_k + {\mathbf D_k^{\operatorname{H}}}\mathbf {\hat W}{\mathbf F^{\operatorname{H}}}\Theta^{\operatorname{H}} \mathbf R_k \notag\\
    &+ {\mathbf R_k^{\operatorname{H}}}\Theta \mathbf F\mathbf {\hat W}\mathbf D_k + {\mathbf R_k^{\operatorname{H}}}\Theta \mathbf F\mathbf {\hat W}{\mathbf F^{\operatorname{H}}}\Theta^{\operatorname{H}} \mathbf R_k. \label{e34}
\end{align} 
Then, by substituting them into Sub-problem \ref{eq30}, and defining ${f_4}\left( \rho^{\left(\operatorname{t+1}\right)},{\mathbf W^{\left(\operatorname{t+1}\right)},\Theta,\mathbf L^{\left(\operatorname{t+1}\right)} } \right)\triangleq f_6\left( \Theta \right)+\operatorname{Consts}\left(\Theta\right)$, where $\operatorname{Consts}\left(\Theta\right)$ denotes the irrelevant constant terms with respect to $\Theta$, i.e., $\sum_{k=1}^K{\operatorname{Tr}\left(\mathbf {\bar U}_k\mathbf D_k^{\operatorname{H}}\mathbf W_k \mathbf L_k^{\operatorname{H}}\right)}$, $\sum_{k=1}^K{\operatorname{Tr}\left(\mathbf {\bar U}_k \mathbf L_k \mathbf W_k^{\operatorname{H}}\mathbf D_k\right)}$, $\sum_{k=1}^K{\operatorname{Tr}\left(\mathbf {\bar U}_k\mathbf L_k   \mathbf D_k^{\operatorname{H}}{\mathbf {\hat W}}\mathbf D_k\mathbf L_k^{\operatorname{H}}\right)}$ and ${\sum _{k = 1}^K {\operatorname{Tr} \left( { {{\mathbf {\bar U} _{k}}}\mathbf L_k\left(\sigma _k^2\mathbf I_{M_u} + {\delta _k^2}/{\rho _k}\mathbf I_{M_u}\right)\mathbf L_k^{\operatorname{H}}}\right)}}$, which have no impact on the update of $\Theta$, and where $f_6\left( \Theta \right)$ can be simplified shown as follows
\begin{align}
    f_6\left( \Theta \right)&= \sum_{k=1}^K{\operatorname{Tr}\left(\mathbf {\bar U}_k\mathbf R_k^{\operatorname{H}}\Theta\mathbf F\mathbf W_k \mathbf L_k^{\operatorname{H}}\right)}\notag\\
    &+\sum_{k=1}^K{\operatorname{Tr}\left(\mathbf {\bar U}_k\mathbf L_k\mathbf W_k^{\operatorname{H}}\mathbf F^{\operatorname{H}}\Theta^{\operatorname{H}}\mathbf R_k\right)}\notag\\
    &-\sum_{k=1}^K{\operatorname{Tr}\left(\mathbf {\bar U}_k\mathbf L_k   {\mathbf D_k^{\operatorname{H}}}\mathbf {\hat W}{\mathbf F^{\operatorname{H}}}\Theta^{\operatorname{H}} \mathbf R_k\mathbf L_k^{\operatorname{H}}\right)}\notag\\
    &-\sum_{k=1}^K{\operatorname{Tr}\left(\mathbf {\bar U}_k\mathbf L_k   {\mathbf R_k^{\operatorname{H}}}\Theta \mathbf F\mathbf {\hat W}\mathbf D_k\mathbf L_k^{\operatorname{H}}\right)}\notag\\
    &-\sum_{k=1}^K{\operatorname{Tr}\left(\mathbf {\bar U}_k\mathbf L_k   {\mathbf R_k^{\operatorname{H}}}\Theta \mathbf F\mathbf {\hat W}{\mathbf F^{\operatorname{H}}}\Theta^{\operatorname{H}} \mathbf R_k\mathbf L_k^{\operatorname{H}}\right)}.
\end{align}
Consequently, we have the following problem, which is equivalent with Problem \ref{eq30}
\begin{align}\label{eqthetaori}
     \mathop{\max} _\Theta \quad&{f_6}\left(\Theta\right)\\
    \operatorname{s.t.} \quad&\eqref{4b}-\eqref{4c}.\notag
\end{align}
To obtain the solution of $\Theta$, we transform Problem \ref{eqthetaori} to an equivalent QP form under constant modulus constraints. First, by defining $\mathbf{T}_{k}=\mathbf R_k\mathbf L_k^{\operatorname{H}}\mathbf {\bar U}_k{\mathbf L_k}{\mathbf R_k^{\operatorname{H}}}$, $\mathbf{Q}=\mathbf F\mathbf {\hat W}{\mathbf F^{\operatorname{H}}}$, we have
\begin{align}
    \operatorname{Tr}\left( {{\mathbf {\bar U}_k}{\mathbf L_k}{\mathbf R_k^{\operatorname{H}}}\Theta \mathbf F\mathbf {\hat W}{\mathbf F^{\operatorname{H}}}{\Theta ^{\operatorname{H}}}\mathbf R_k \mathbf L_k^{\operatorname{H}}} \right) 
    \triangleq \operatorname{Tr}\left( {{\Theta ^{\operatorname{H}}}{\mathbf T_k}\Theta \mathbf Q} \right).
\end{align}
Then, by defining $\mathbf A_{k}=\mathbf F \mathbf {\hat W} \mathbf D_k\mathbf L_k^{\operatorname{H}}\mathbf {\bar U}_k\mathbf L_k\mathbf R_k^{\operatorname{H}}$, we have
\begin{align}
    \operatorname{Tr}\left( {{\mathbf {\bar U}_k}{\mathbf L_k}{\mathbf D_k^{\operatorname{H}}}\mathbf {\hat W}{\mathbf F^{\operatorname{H}}}{\Theta ^{\operatorname{H}}}\mathbf R_k \mathbf L_k^{\operatorname{H}}} \right) 
    \triangleq \operatorname{Tr}\left( {{\Theta ^{\operatorname{H}}}{\mathbf A_k^{\operatorname{H}}}} \right).
\end{align}
Next, by defining ${\mathbf Z_{k}}=\mathbf F{\mathbf W_k}{\mathbf {\bar U}_k}{\mathbf L_k}{\mathbf R_k^{\operatorname{H}}}$, we have
\begin{align}
    \operatorname{Tr}\left( {{\mathbf {\bar U}_k}{\mathbf R_k^{\operatorname{H}}}\Theta \mathbf F{\mathbf W_k}{\mathbf L_k^{\operatorname{H}}}} \right) 
    \triangleq \operatorname{Tr}\left( {\Theta {\mathbf Z_k}} \right).
\end{align}
Similar, by defining $\mathbf J_k=\mathbf R_k\mathbf R_k^{\operatorname{H}}$, $\Lambda _k = \mathbf F\mathbf {\hat W} \mathbf D_k\mathbf R_k^{\operatorname{H}}$, the constraint \eqref{4b} can be expressed as 
\begin{align}
    \operatorname{Tr}\left( {{\Theta ^{\operatorname{H}}}{\mathbf J_k}\Theta \mathbf Q} \right) + \operatorname{Tr}\left( {{\Theta ^{\operatorname{H}}}\Lambda _k^{\operatorname{H}}} \right) + \operatorname{Tr}\left( {\Theta {\Lambda _k}} \right) \notag\\
    \ge \mathbf {\dot e}_k - \operatorname{Tr}\left( {\mathbf D_k^{\operatorname{H}}\mathbf {\hat W}{\mathbf D_k}} \right) \triangleq \mathbf {\dddot e}_k.\label{e43}
\end{align}
Therefore, by defining $\mathbf T = {\sum_{k = 1}^K {{\mathbf T_{k}}} } $, $\mathbf C = {\sum_{k = 1}^K {{\mathbf A_{k}}-{\mathbf Z_{k}}} } $, an equivalent problem is defined as
\begin{align}\label{eq37}
    \mathop {\min }_\Theta \quad & \operatorname{Tr}\left( {{\Theta ^{\operatorname{H}}}\mathbf T\Theta \mathbf Q} \right) + \operatorname{Tr}\left( {{\Theta ^{\operatorname{H}}}{\mathbf C}^{\operatorname{H}}} \right) + \operatorname{Tr}\left( {\Theta \mathbf C} \right) \\
    \operatorname{s.t.} \quad & \eqref{e43},\;\eqref{4c}.\notag
\end{align}
After invoking some further matrix operations \cite{Ben2012Lectures}, and by defining $\theta = \left( \Theta _{1,1},\Theta _{2,2},\cdots,\Theta _{N,N}  \right)^{\operatorname{T}}$, $\Omega\triangleq{\mathbf T \odot {\mathbf Q^{\operatorname{T}}}}$ and $ \mathbf c = \left(  \mathbf C _{1,1},\mathbf C _{2,2},\cdots,\mathbf C _{N,N} \right)^{\operatorname{T}}$, we arrive at $\operatorname{Tr}\left( {{\Theta ^{\operatorname{H}}}\mathbf T\Theta \mathbf Q} \right)\triangleq {\theta^{\operatorname{H}}}\Omega \theta$, $\operatorname{Tr}\left( {\Theta^{\operatorname{H}} \mathbf C^{\operatorname{H}}} \right){ = }{\theta^{\operatorname{H}}}\mathbf c^*$ and $\operatorname{Tr}\left( {\Theta \mathbf C} \right) = {\mathbf c^{\operatorname{T}}}\theta$. 
Similarly, the set of constraints \eqref{e43} are equivalent to 
\begin{align}
    {\theta ^{\operatorname{H}}}{\mathbf {\bar J}_k}\theta  + 2\Re \left\{ {{\theta ^{\operatorname{H}}}{\lambda_k ^ * }} \right\} \ge {\mathbf {\dddot e}_k}, \forall k,\label{e46}
\end{align}
where $\lambda _k= {\left( \left[ \Lambda _k \right]_{1,1},\left[ \Lambda _k \right]_{2,2},\cdots,\left[ \Lambda _k \right]_{N,N}\right)^{\operatorname{T}}}$ and ${\mathbf {\bar J}_k} = {\mathbf J_k} \odot {\mathbf Q^{\operatorname{T}}}$. It can be checked that $\mathbf {\bar J}_k, \forall k$ and $\Omega$ are semi-definite matrices.
As for constraint \eqref{4c}, which can be written as 
\begin{align}
    \left|\theta_n\right|  = \frac{1}{\alpha}, \forall n. \label{e48}
\end{align}
Then, we arrive at
\begin{align}
    \mathop {\min }_\theta & \quad f_7\left(\theta\right)\triangleq{\theta ^{\operatorname{H}}}\Omega \theta  + 2\Re \left\{ {{\theta ^{\operatorname{H}}}{c^ * }} \right\} \label{eqbeforesca}\\
    \operatorname{s.t.}  &\quad \eqref{e46},\;\eqref{e48}.\notag 
\end{align}

In $q$-th sub-iterations for updating PRB, similar to the operation in \eqref{e26}, we can approximate \eqref{e46} to its first-order Taylor expansion, which leads to ${\theta ^{\operatorname{H}}}{\mathbf {\bar J}_k}\theta  \ge  - \left({\theta ^{\left( {q} \right)}}\right)^{\operatorname{H}}{\mathbf {\bar J}_k}{\theta ^{\left( {q } \right)}} + 2\Re \left\{ {{\theta ^{\operatorname{H}}}{\mathbf {\bar J}_k}{\theta ^{\left( {q} \right)}}} \right\}$. Hence, we have 
\begin{align}
    2\Re \left\{ {{\theta ^{\operatorname{H}}}\left( {{\lambda_k ^ * } + {\mathbf {\bar J}_k}{\theta ^{\left( {q} \right)}}} \right)} \right\} \ge {\mathbf {\dddot e}_k} + \left({\theta ^{\left( {q} \right)}}\right)^{\operatorname{H}}{\mathbf {\bar J}_k}{\theta ^{\left( {q} \right)}} \triangleq \mathbf{\hat e}_k^{\left( {q} \right)}.\label{e47}
\end{align}
Therefore, the problem becomes
\begin{align}
    \mathop {\min }_\theta & \quad f_7\left(\theta\right) \label{e49}\\
    \operatorname{s.t.}  &\quad \eqref{e48},\;\eqref{e47}.\notag 
\end{align}
Although the objective function $f_7\left(\theta\right)$ is convex and the constraint \eqref{e47} is linear, Problem \ref{e49} is still a non-convex QP problem due to the constant modulus constraints in \eqref{e48}. 

Here, we handle the above problem based on the well-known MM approach to obtain optimal PRB in a nearly closed-form with low complexity. 
The key to the success of MM algorithm lies in constructing a sequence of convex surrogate functions as the bound of $f_7\left(\theta\right)$ in Problem \ref{e49}. As a benefit, we design the surrogate function by employing the second-order Taylor expansion, which can replace $\Omega$ by desired structures (e.g., diagonal matrix). Hence, we have the following inequality
\begin{align}
    {\theta ^{\operatorname{H}}}\Omega \theta  \le {\theta ^{\operatorname{H}}}\Xi \theta  &- 2\Re \left\{ {{\theta ^{\operatorname{H}}}\left( {\Xi  - \Omega } \right){\theta ^{\left( {q} \right)}}} \right\} \notag\\
    &+ \left({\theta ^{\left( {q} \right)}}\right)^{\operatorname{H}}\left( {\Xi  - \Omega } \right){\theta ^{\left( {q} \right)}},\label{e50}
\end{align}
where $\Xi  = {\omega _{\max }}\mathbf I_{N}$ with $\omega _{\max }$ denotes the maximum eigenvalue of $\Omega$. Clearly, the inequality \eqref{e50} holds when $\Xi \succeq \Omega$ and the equality is achieved at $\theta=\theta^{\left(q\right)}$. Hence, we arrive that
\begin{align}
    g\left( {\theta | \theta^{\left( {q} \right)} } \right) 
    =\theta ^{\operatorname{H}}\Xi \theta&+2\Re \left \{ \theta ^{\operatorname{H}}\left( {\Omega  \theta^{\left( {q} \right)}  - \Xi  \theta ^{\left( {q} \right)} + {\mathbf c^ * }} \right)\right \}\notag\\
    &+ {{ \left(\theta^{\left( {q} \right)} \right)}^{\operatorname{H}}}\left( {\Xi  - \Omega } \right) \theta^{\left( {q} \right)} .
\end{align}
Clearly, $g\left(\theta| \theta ^{\left( {q } \right)}\right)$ satisfies the following conditions: $\left(a\right)$\;$g\left( {\theta | \theta^{\left( {q} \right)} } \right)$ is an upper-bounded function of $f_7\left(\theta \right)$; $\left(b\right)$\;$g\left( {\theta  | \theta^{\left( {q} \right)} } \right)$ and $f_7\left(\theta\right)$ have the same solution when $\theta=\theta^{\left( {q} \right)}$; $\left(c\right)$\;$g\left( {\theta  | \theta^{\left( {q} \right)} } \right)$ and $f_7\left(\theta \right)$ achieve a same gradient at $\theta^{\left( {q} \right)}$.
The proof can be obtained by employing Example 13 in \cite{MM}. 
Since $\theta^{\operatorname{H}} \theta=\frac{N}{\alpha}$ holds, we have $\theta^{\operatorname{H}}\Xi\theta=\frac{N\omega _{\max}}{\alpha}$, which is a constant with respect to $\theta$. 

Therefore, by defining $\mathbf r^{\left(q\right)} = \left(\omega_{\max}\mathbf I_{N}-\Omega \right)\theta^{\left(q\right)}-\mathbf c^*$ and omitting the irrelevant constant terms with respect to $\theta$, i.e., ${{ \left(\theta^{\left( {q} \right)} \right)}^{\operatorname{H}}}\left( {\Xi  - \Omega } \right) \theta^{\left( {q} \right)}$ and $\theta^{\operatorname{H}}\Xi\theta$, we have the following problem
\begin{align}
    \mathop {\max }_\theta & \quad 2\Re \left\{{\theta ^{\operatorname{H}}\mathbf r^{\left(q\right)}}\right\} \label{e52}\\
    \operatorname{s.t.}  &\quad \eqref{e48}, \eqref{e47}.\notag
\end{align}
However, due to the non-convex constant modulus constraints in \eqref{e48}, the dual gap is not guaranteed to be zero, and the Lagrangian dual method cannot be used directly. In following, we provide a penalty-based method. It is worth noting that there always exists a vector $\chi=[\chi_1,\chi_2,\cdots,\chi_K]$, each element of which satisfies $\chi_k \ge 0 $, and Problem \ref{e52} can be reformulated as
\begin{align}
    \mathop {\max }_\theta & \quad 2\Re \left\{{\theta ^{\operatorname{H}}\mathbf r^{\left(q\right)}}\right\}+2\sum_{k=1}^K{\chi_k\Re \left\{ {{\theta ^{\operatorname{H}}}\left( {{\lambda_k ^ * } + {\mathbf {\bar J}_k}{\theta ^{\left( {q} \right)}}} \right)} \right\}} \label{e53}\\
    \operatorname{s.t.}  &\quad \eqref{e48}.\notag
\end{align}
It can be checked that Problem \ref{e52} and \ref{e53} are  equivalent, and have the same optimal solution. Interestingly, the objective function and the constraints in Problem \ref{e53} are separable with respect to $\theta_n,\forall n$. Therefore, we can solve Problem \ref{e53} by solving $N$ separate problems in parallel. Let $\mathbf f^{\left({q}\right)}=r^{\left(q\right)}+\sum_{k=1}^{K}{\chi_k\left({{\lambda_k ^ * } + {\mathbf {\bar J}_k}{\theta ^{\left( {q} \right)}}}\right)}$, clearly, 
Problem \ref{e53} can be further rewritten as 
\begin{align}\label{eq47}
    \mathop {\max }_{\arg\theta_n} &  \cos{\left(\arg\theta_n-\arg \mathbf f^{\left({q}\right)}_n\right)}, \forall n.
\end{align}
Problem \ref{eq47} is maximized when the phases of $\theta_n$ and $\mathbf f^{\left({q}\right)}_n$ are equal. Therefore, the optimal solution of Problem \ref{eq47} can be sequentially obtained as following
\begin{align}\label{e54}
    \arg\theta^{\left({q+1}\right)}_n  = \arg  \mathbf f _n^{\left({q}\right)}, \forall n .
\end{align}
Meanwhile, the penalty parameter of $\chi^{\left(q+1\right)}$ can be updated by using the ellipsoid method \cite{2004Convex}.

The details of optimizing PRB is summarized in Algorithm \ref{a2}. It can be verified that the sequence of the solutions generated by Algorithm \ref{a2}, i.e., $\left\{\theta^{\left(q\right)}, q=1,2,\cdots,\hat q\right\}$, are the optimal solutions of Problem \ref{e52}, where $\hat q$ represents the maximum number of iterations when algorithm converges. 
Consequently, we have the following lemma
\begin{lemma}\label{lemma5}
The final converged solution of Algorithm \ref{a2}, i.e., $\theta ^{\left(\hat q\right)}$, satisfies the KKT conditions of Problem \ref{eqbeforesca}.
\end{lemma}
\begin{proof}
The proof is presented in Appendix \ref{applemma5}.
\end{proof}

\begin{algorithm}[t] 
    \caption{Penalty-based MM/SCA Algorithm for Optimizing PRB} 
\label{a2} 
\begin{algorithmic}[1]
\REQUIRE $\theta ^{\left ( t \right )}$, threshold $\varepsilon_2$.
\ENSURE Optimal value of $\theta^{\left(t+1\right)}$
\STATE \textbf{Set} $\theta ^{\left ( 0 \right )}\triangleq\theta ^{\left ( t \right )}$;
\STATE \textbf{Calculate} $f_7\left(\theta^{\left ( 0 \right )}\right)=\left({\theta ^{\left( {0} \right)}}\right)^{\operatorname{H}}\Omega \theta^{\left ( 0 \right )}  + 2\Re \left\{ {{\left({\theta ^{\left( {0} \right)}}\right)^{\operatorname{H}}}{c^ * }} \right\}$;\\
\FOR{$q=1 \;\text{to}\; 2,3,\cdots,q_{\max}$}
\STATE \textbf{Update} $\mathbf{\hat e}_k^{\left( {q} \right)}={\mathbf {\dddot e}_k^{\left( {q} \right)}} + {\left({\theta ^{\left( {q} \right)}}\right)^{\operatorname{H}}}{\mathbf {\bar J}_k}{\theta ^{\left( {q} \right)}}$, by using \eqref{e47};
\STATE \textbf{Update} $\theta ^{\left ( q+1 \right )}$, by using \eqref{e54};
\IF {${\left|f_7\left(\theta^{\left(q+1\right)}\right) - f_7\left(\theta^{\left(q\right)}\right) \right|}/{f_7\left(\theta^{\left(q\right)}\right)} < \varepsilon_2 $} 
\item {Break\;;}
\ENDIF
\ENDFOR
\STATE {$\theta^{\left(t+1\right)}\triangleq\theta^{\left(\hat q\right)}$}.
\end{algorithmic} 
\end{algorithm}


\subsection{Convergence and Complexity of the Overall Algorithm}

The details of the proposed AO-based joint PS ratios, ATB and PRB optimization (JPSAPBO) algorithm are provided in Algorithm \ref{a3}. 
The proposed algorithm is guaranteed to converge. Particularly, based on the fact of LDT and FPT methods, we have 
\begin{align}
    \mathcal{R}\left( {{\rho, \mathbf W},{\Theta }} \right) &= {f_1}\left( {\rho, {\mathbf W},{\Theta},{\mathbf U}} \right)\notag\\
    &={f_3}\left( {\rho, {\mathbf W},{\Theta },{\mathbf U}, \mathbf L} \right).
\end{align}
In $t${-th} iteration, we have the following inequalities
\begin{align}
    &{f_3}\left( {\rho^{\left(t\right)}, {\mathbf W^{\left(t\right)}},{\Theta^{\left(t\right)} },{\mathbf U}^{\left(t\right)}, \mathbf L^{\left(t\right)}} \right) \notag\\
    \mathop  \le^{\left( a \right)} & {f_3}\left( {\rho^{\left(t\right)}, {\mathbf W^{\left(t\right)}},{\Theta^{\left(t\right)} },{\mathbf U}^{\left(t+1\right)}, \mathbf L^{\left(t\right)}} \right)\notag\\
    \triangleq&{f_4}\left( {\rho^{\left(t\right)}, {\mathbf W^{\left(t\right)}},{\Theta^{\left(t\right)} }, \mathbf L^{\left(t\right)}} \right) + \operatorname{Const}\left(\mathbf {\bar U}\right)\notag\\
    \mathop  \le^{\left( b \right)}&{f_4}\left( {\rho^{\left(t+1\right)}, {\mathbf W^{\left(t\right)}},{\Theta^{\left(t\right)} }, \mathbf L^{\left(t\right)}} \right) + \operatorname{Const}\left(\mathbf {\bar U}\right)\notag\\
    \mathop  \le^{\left( c \right)}&{f_4}\left( {\rho^{\left(t+1\right)}, {\mathbf W^{\left(t\right)}},{\Theta^{\left(t\right)} }, \mathbf L^{\left(t+1\right)}} \right) + \operatorname{Const}\left(\mathbf {\bar U}\right)\notag\\
    \triangleq&-{f_5}\left({\mathbf W^{\left(t\right)}} \right) + \operatorname{Const}\left( {\rho},\mathbf {L},\mathbf {\bar U} \right)\notag\\
    \mathop  \le^{\left( d \right)}& -{f_5}\left({\mathbf W^{\left(t+1\right)}} \right) + \operatorname{Const}\left( {\rho},\mathbf {L},\mathbf {\bar U} \right).
\end{align}
where the inequalities, e.g., $\left( a \right)$ holds based on the fact that the optimality of $\mathbf U^{\left(t+1\right)}$ has been proved in \eqref{eq13} when the other variables are fixed. Similarly, the inequalities of $\left( b \right)$ and $\left( c \right)$ for optimizing $\rho$ and $\mathbf L$, respectively, can also be verified. The inequality $\left( d \right)$ holds by employing Algorithm \ref{a1}.

Moreover, based on the conditions of MM approach, we know that $g\left(\theta^{\left( {t+1 } \right)}| \theta ^{\left( {t } \right)}\right) \le g\left(\theta^{\left( {t } \right)}| \theta ^{\left( {t } \right)}\right)$, $g\left(\theta^{\left( {t } \right)}| \theta ^{\left( {t } \right)}\right) = f_7\left(\theta^{\left( {t } \right)} \right)$, $g\left(\theta^{\left( {t+1 } \right)}| \theta ^{\left( {t } \right)}\right) \ge f_7\left(\theta^{\left( {t+1 } \right)} \right)$, and $f_7\left(\theta^{\left( {t+1 } \right)} \right)\le f_7\left(\theta^{\left( {t } \right)} \right)$. Hence, we have the following inequality
\begin{align}
    &{f_3}\left( {\rho^{\left(t+1\right)}, {\mathbf W^{\left(t+1\right)}},{\Theta^{\left(t\right)} },{\mathbf U}^{\left(t+1\right)}, \mathbf L^{\left(t+1\right)}} \right) \notag\\
    \triangleq& {f_6}\left({ \Theta^{\left(t\right)}} \right) + \operatorname{Const}\left( {\rho},\mathbf W,\mathbf {\bar U},\mathbf {L} \right)\notag\\
    \triangleq&-{f_7}\left({ \theta^{\left(t\right)}} \right) + \operatorname{Const}\left( {\rho},\mathbf W,\mathbf {\bar U},\mathbf {L} \right)\notag\\
    \le& -{f_7}\left({ \theta^{\left(t+1\right)}} \right) + \operatorname{Const}\left( {\rho},\mathbf W,\mathbf {\bar U},\mathbf {L} \right).
\end{align}
Above inequalities verify that $\mathcal{R}\left( {{\rho, \mathbf W},{\Theta }} \right)$ is monotonically non-decreasing after each updating step. In short, it can be concluded that
\begin{align}\label{eq56}
    \mathcal{R}\left( {{\rho^{\left(t+1\right)}, \mathbf W^{\left(t+1\right)}},{\Theta^{\left(t+1\right)} }} \right)\ge \mathcal{R}\left( {{\rho^{\left(t\right)}, \mathbf W^{\left(t\right)}},{\Theta^{\left(t\right)} }} \right)\notag\\
    \ge \cdots \ge \mathcal{R}\left( {{\rho^{\left(1\right)}, \mathbf W^{\left(1\right)}},{\Theta^{\left(1\right)} }} \right).
\end{align}
The inequality in \eqref{eq56} investigate that $\mathcal{R}\left( {{\rho, \mathbf W},{\Theta }} \right)$ is monotonically non-decreasing over iterations.
In addition, because of the constraints in \eqref{4a}-\eqref{4d},  the data rate $\mathcal{R}\left( {{\rho, \mathbf W},{\Theta }} \right)$ has an upper bound. As the number of iterations increases, we have $\mathcal{R}\left( {{\rho^{\left(\hat t\right)}, \mathbf W^{\left(\hat t\right)}},{\Theta^{\left(\hat t\right)} }} \right) \triangleq  \mathcal{R}\left( {{\rho^{\operatorname{opt}}, \mathbf W^{\operatorname{opt}}},{\Theta^{\operatorname{opt}}}} \right)$, where $\hat t$ denotes the maximum number of iterations when Algorithm \ref{a3} converges. Therefore, we complete the proof of the strict convergence of Algorithm \ref{a3}.

Meanwhile, the complexity analysis of Algorithm \ref{a3} is given as below.
In step 4-6, the complexities of updating $\mathbf U_k$, $\rho _k$ and $\mathbf L_k$ are $\mathcal O \left(M_u^3\right)$, $\forall k$, respectively. In step 7, the complexity of calculating $\mathbf W _k, \; \forall k$ by employing Algorithm \ref{a1} is $\mathcal O \left(\mathcal I_{1}M_b^3\right)$, where $\mathcal I_{1}$ denotes the number of sub-iterations when Algorithm \ref{a1} converges. In step 8, the complexity of calculating the maximum eigenvalue of $\Omega$ is $\mathcal O \left(N^3\right)$ and the complexity of updating $\theta$ is $\mathcal O \left(\mathcal I_{2} N^2\right)$, where $\mathcal I_{2}$ stands for the number of sub-iterations when Algorithm \ref{a2} converges. Therefore, the total complexity of Algorithm \ref{a3} is $\mathcal O \left(\mathcal I_{total}\left(K\left(3M_u^3+\mathcal I_{1}M_b^3\right)+\mathcal I_{2} N^2+N^3\right)\right)$, where $\mathcal I_{total}$ is the number of iterations when Algorithm \ref{a3} converges.

\begin{algorithm}[t]
\caption{AO-based JPSAPBO Algorithm} 
\label{a3} 
\begin{algorithmic}[1]
\REQUIRE    
            $\mathbf \rho_{k} ^{\left ( 1 \right )},\forall k$,
            $\mathbf W_{k} ^{\left ( 1 \right )},\forall k$,
             $\Theta ^{\left ( 1 \right )}$;
            threshold $\varepsilon $.
\ENSURE     Optimal value of $ \mathbf \rho_{k}^{\operatorname{opt}}$, $ \mathbf W_{k}^{\operatorname{opt}},\forall k$; $ \Theta^{\operatorname{opt}}.$\\
\STATE\textbf{Calculate} $\Gamma_{k} ^{\left ( 1 \right )}, \forall\, k$;\\
\STATE\textbf{Calculate} $\mathcal{R}^{\left( 1 \right )}\left(\rho^{\left ( 1 \right )},\mathbf W^{\left ( 1 \right )},\Theta^{\left ( 1 \right )}\right)$;\\
\FOR{$t=1 \;\text{to}\; 2,3,\cdots,t_{\max}$ }
\STATE \textbf{Update} $\mathbf U_{k} ^{\left ( t+1 \right )}, \forall\, k$, by using \eqref{eq14};
\STATE \textbf{Update} $\mathbf \rho_{k} ^{\left ( t+1 \right )}, \forall\, k$, by using \eqref{eq17};
\STATE \textbf{Update} $\mathbf L_{k} ^{\left ( t+1 \right )}, \forall\, k$, by using \eqref{eq18};
\STATE \textbf{Update} $\mathbf W_{k} ^{\left ( t+1 \right )}, \forall\, k$, by employing Algorithm \ref{a1};
\STATE \textbf{Update} $\Theta ^{\left ( t+1 \right )}$, by employing Algorithm \ref{a2};
 \IF {$\frac{{\left| {{\mathcal{R}^{\left( t+1 \right)}}\left( {\rho,\mathbf W,\Theta } \right) - {\mathcal{R}^{\left( {t} \right)}}\left( {\rho,\mathbf W,\Theta } \right)} \right|}}{{{\mathcal{R}^{\left( t \right)}}\left( {\rho,\mathbf W,\Theta } \right)}} < \varepsilon_3 $} 
\item Break\;;
 \ENDIF
 \ENDFOR
 \STATE $ \mathbf \rho_{k}^{\operatorname{opt}} \triangleq \mathbf \rho_{k}  ^{\left ( \hat t \right )}$,$ \mathbf W_{k}^{\operatorname{opt}} \triangleq \mathbf W_{k}  ^{\left (\hat t \right )},\forall k$; $ \Theta^{\operatorname{opt}} \triangleq \Theta  ^{\left ( \hat t \right )}.$
\end{algorithmic} 
\end{algorithm}
\section{numerical simulation results}
In this section, numerical simulation results are provided to evaluate the performance of the proposed AO-based JPSAPBO algorithm.  
It is assumed that there is a uniform linear array (ULA) at the AP side, and there is a uniform rectangular array (URA) at the IRS side.
The total number of reflective elements of the IRS is assumed to be $N=N_xN_y$, where $N_h$ and $N_v$ denote the numbers of reflecting elements along horizon and vertical, respectively. 
Since the AP, the IRS and the PSRs are closed to each other, the small-scale channel fading is assumed to be Rician fading, and the channels, i.e., the AP-IRS channel, the AP-PSRs channels and the IRS-PSRs channels can be mathematically modeled as
\begin{align}
    {\mathbf {\hat H}_i} =  {\sqrt {\frac{{{\beta_i }}}{{{\beta_i } + 1}}} {\mathbf {\hat H}_i}^{\operatorname{LOS}}  + \sqrt {\frac{1}{{{\beta_i } + 1}}} {\mathbf {\hat H}_i}^{\operatorname{NLOS}}},
\end{align}
where $\beta_i$ stands for the Rician factor, and ${\mathbf {\hat H}_i}^{\operatorname{LOS}}$ is the deterministic line of sight (LOS), 
$\mathbf{\hat H}_i^{\operatorname{NLOS}}$ is the non-line-of-sight (NLOS) component, which follow Rayleigh fading model, with $i\in\operatorname{\left\{AP-IRS, AP-PSRs, IRS-PSRs\right\}}.$ 
Meanwhile, the distance-dependent large-scale path loss is given as
\begin{align}
    \mathcal{L}\left( {{d}} \right) = {C_0}{\left( {\frac{{{d}}}{{{D_0}}}} \right)^{ - \chi_j }},
\end{align}
where $C_0$ corresponds to the path loss at the reference distance of $D_0=1$m, $d$ is the link distance and $\chi_j$ denotes path loss exponent, where $j\in\operatorname{\left\{direct, relate\right\}}$. 
The other simulation parameters are summarized in Table \ref{tabaaa}, unless otherwise stated.

\begin{table}[t]
    \centering
    \caption{Simulation Parameters}
    \begin{tabular}{c|c}
    \hline 
    \textbf{Parameters}&\textbf{Values}  \\ \hline
    Cell coverage &5 m\\ \hline
    AP location &$\left(0\operatorname{m},0\operatorname{m}\right)$\\\hline
    IRS location &$\left(5\operatorname{m},5\operatorname{m}\right)$\\\hline
    Center of PSRs location &$\left(5\operatorname{m},0\operatorname{m}\right)$\\\hline
    Number of antennas of the AP&$M_b=8$  \\ \hline
    Number of antennas of each PSR&$M_u=2$  \\ \hline
    Number of PSRs&$K=4$  \\ \hline
    Number of reflective elements&$N=30$  \\ \hline
    Path loss at the reference distance&$C_0=-30$ dB  \\ \hline
    The maximum transmission power&$P_{\max}=10$ W  \\ \hline
    Path loss exponent of direct channel&$\chi_{\operatorname{direct}}=3.6$ \\ \hline
    Path loss exponent of IRS-related channel&$\chi_{\operatorname{relate}}=2.2$ \\ \hline
    Rician factor of channels&$\beta=5$ dB \\ \hline
    Antenna noise power & $\sigma_k^2=\sigma^2=-50$ dBm \\ \hline
    Signal Processing noise power& $\delta_k^2=\delta^2=-40$ dBm \\ \hline
    Minimum harvested power requirement & $\mathbf{\bar e} _k=\mathbf{\bar e}=0.5$ mW \\ \hline
    Energy conversion efficiency & $\eta_k=\eta=0.7$\\ \hline
    Threshold&$\varepsilon_{1,2,3}=10^{-6}$ \\ \hline
    \end{tabular}
    \label{tabaaa}
\end{table}

\begin{figure}[t]
    \centering
    \scalebox{0.55}{\includegraphics{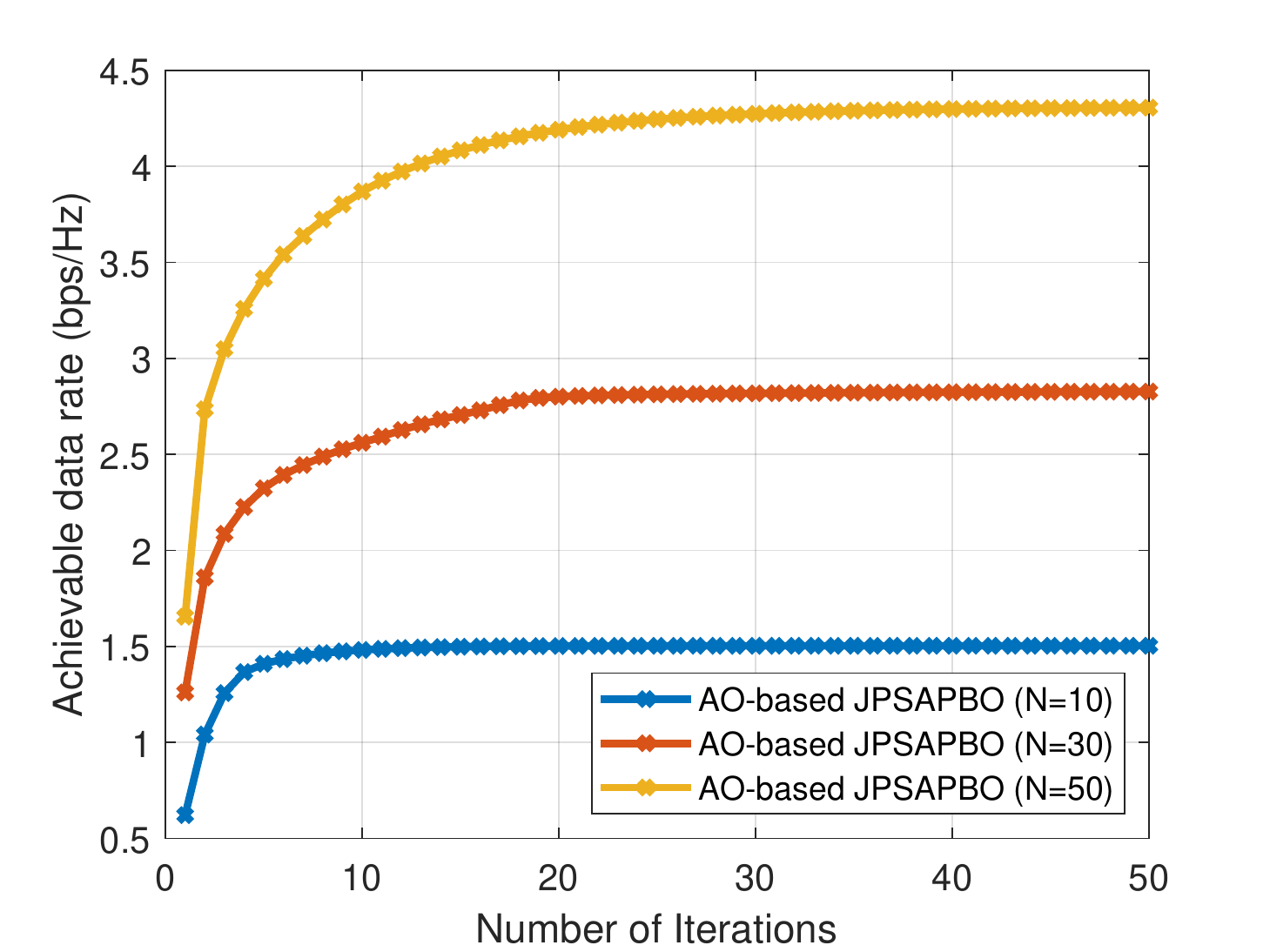}}
    \caption{Convergence behavior of the AO-based JPSAPBO algorithm under different number of reflective elements.}
 \label{fconvergence} 
\end{figure}
Before the performance comparison, we first study the convergence behavior of the proposed AO-based JPSAPBO algorithm in Fig. \ref{fconvergence}. Particularly, we plot the data rate versus the number of iterations for the various number of reflective elements, i.e., $N=10$, $N=30$ and $N=50$.
The curves are consistent to our expectation, as we can observe that the proposed AO-based JPSAPBO algorithm converges to a stationary point after a few iterations.
Besides, another observation is that more reflective elements only leads to a slightly slower convergence speed, e.g., for a large-scale IRS with $N=50$, the proposed AO-based JPSAPBO algorithm can also converge in $30$ iterations.

Then, in the following, we study the performance gain achieved by the proposed AO-based JPSAPBO algorithm. For comparison, we introduce three benchmark system design schemes to validate the performance as following
\begin{itemize}
    \item Fixed PS ratios: In this scheme, the PS ratios of PSRs are fixed (i.e., $\rho_k=0.5,\forall k$) and without optimized, while the ATB and PRB are need to be optimized;
    \item Random Phase Shifts: In this scheme, the phase shifts of IRS elements are random, and only the PS ratios at the PSRs side and ATB at the AP side need to be optimized;
    \item Without IRS: In this scheme, there is no IRS in the system, i.e., it is a conventional SWIPT, and only two sides, i.e., the AP side and the PSRs side need to be optimized. 
\end{itemize}

\begin{figure}[t]
    \centering
    \scalebox{0.55}{\includegraphics{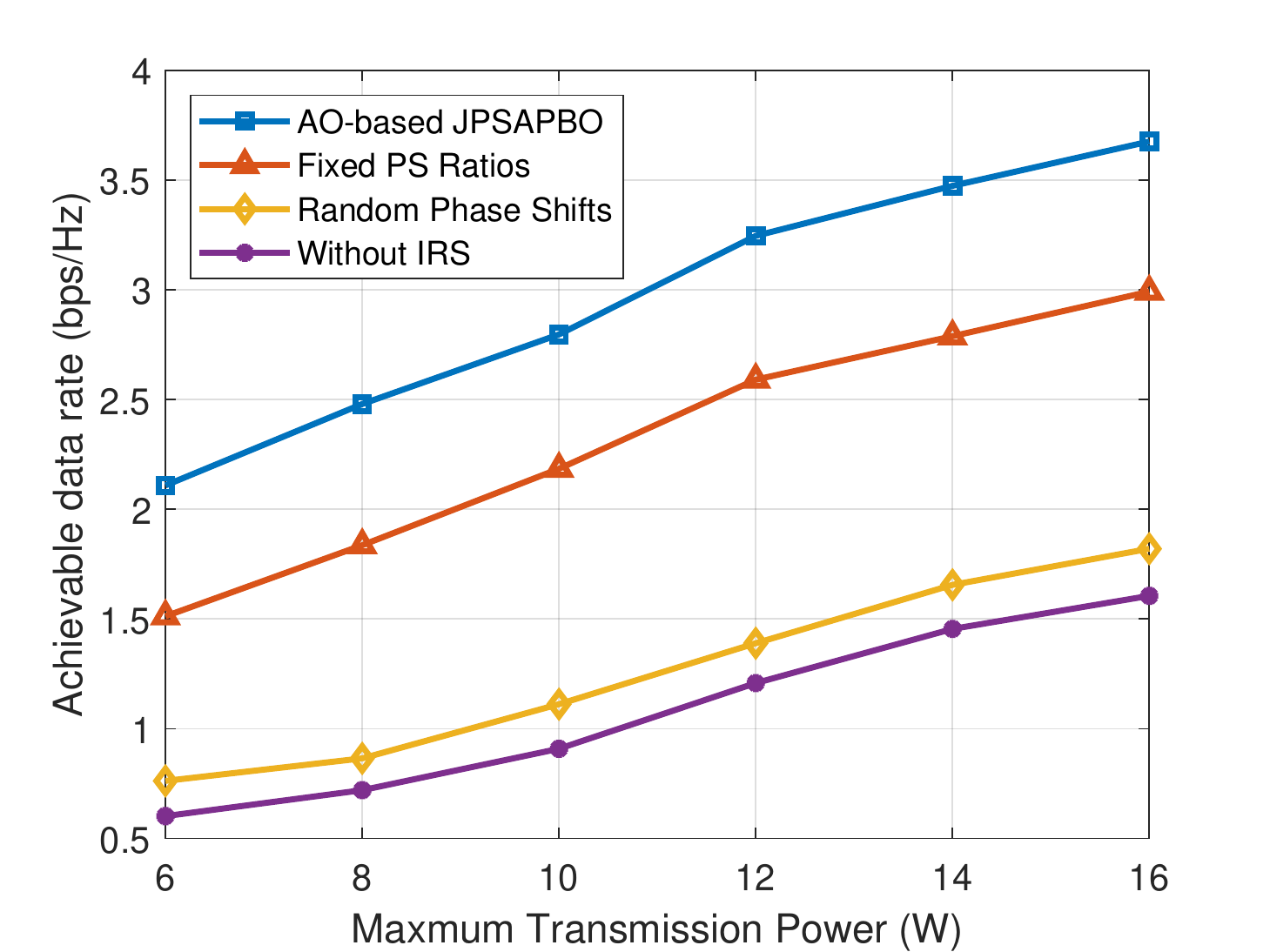}}
    \caption{Achievable data rate versus the maximum transmission power of the AP, i.e., $P_{\max}$.}
 \label{fpower} 
\end{figure}
As shown in Fig. \ref{fpower}, we investigate the data rate versus the maximum transmission power of AP, i.e., $P_{\max}$ for different schemes. 
The proposed AO-based JPSAPBO algorithm achieves a significant performance gain comparing with benchmark schemes. 
Meanwhile, it can be seen that the performance achieved by the Fixed PS ratios scheme is better than both the Random Phase Shifts scheme and the Without IRS scheme, this implies that IRS can enhance the performance of SWIPT system, and PRB of IRS need be carefully optimized to achieved a higher data rate.  

\begin{figure}[t]
    \centering
    \scalebox{0.55}{\includegraphics{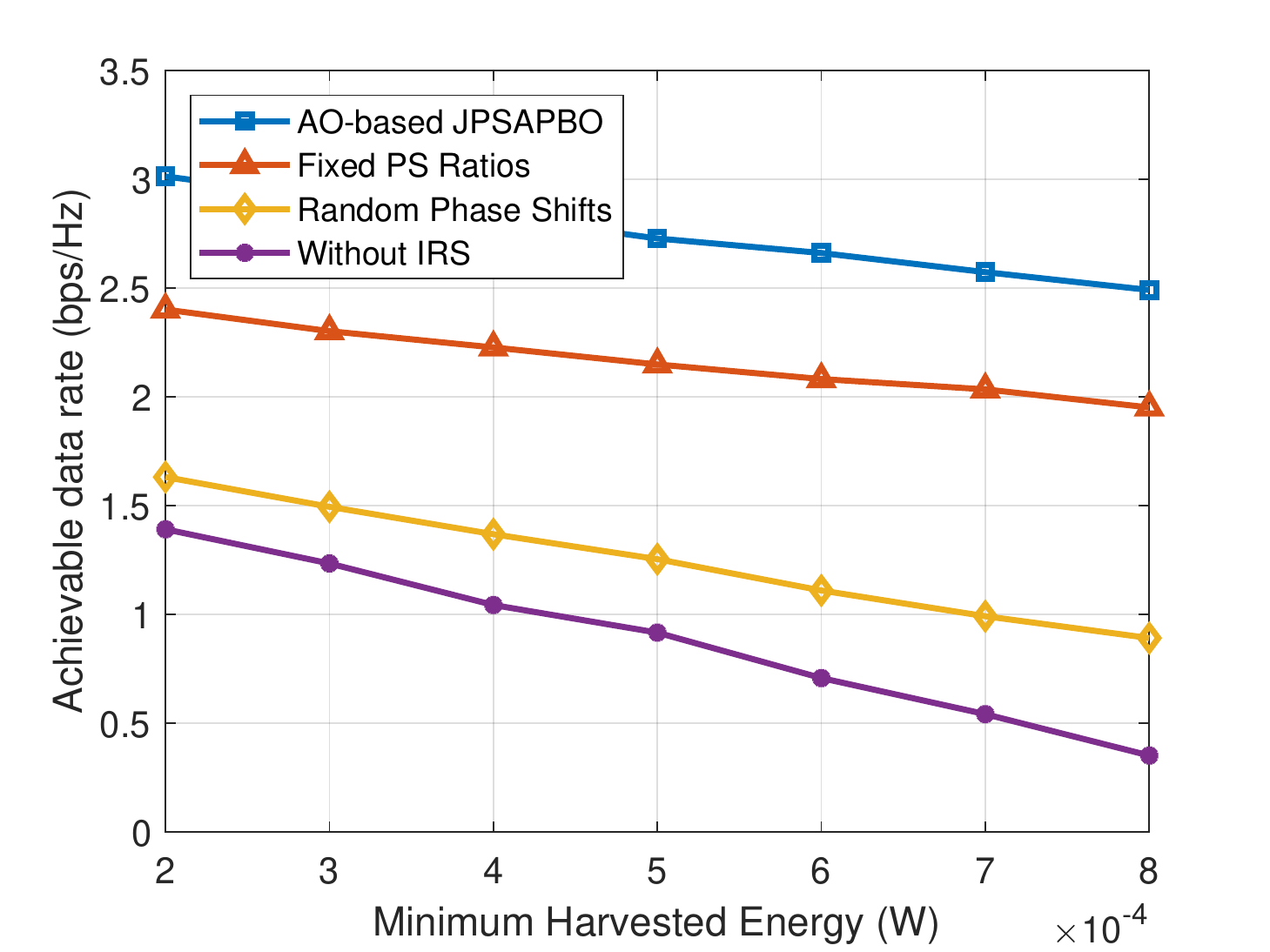}}
    \caption{Achievable data rate versus the minimum harvested energy requirement of PSRs, i.e., $\mathbf {\bar e}_k, \forall k$.}
 \label{fharvestedpower} 
\end{figure}
As shown in Fig. \ref{fharvestedpower} illustrates the data rate achieved by the proposed AO-based JPSAPBO algorithm and benchmark schemes over the minimum harvested power requirement of PSRs, i.e., 
${\mathbf{\bar e}}_k,\forall k$. 
As we can see, the proposed AO-based JPSAPBO algorithm in this paper outperforms the benchmarks significantly. 
In addition, with the increased minimum harvested energy, the gap between the IRS-related schemes and the Without IRS scheme becomes larger, and this indicates that the systems with IRS are more robust against the minimum harvest energy comparing with the system without IRS.

\begin{figure}[t]
    \centering
    \scalebox{0.55}{\includegraphics{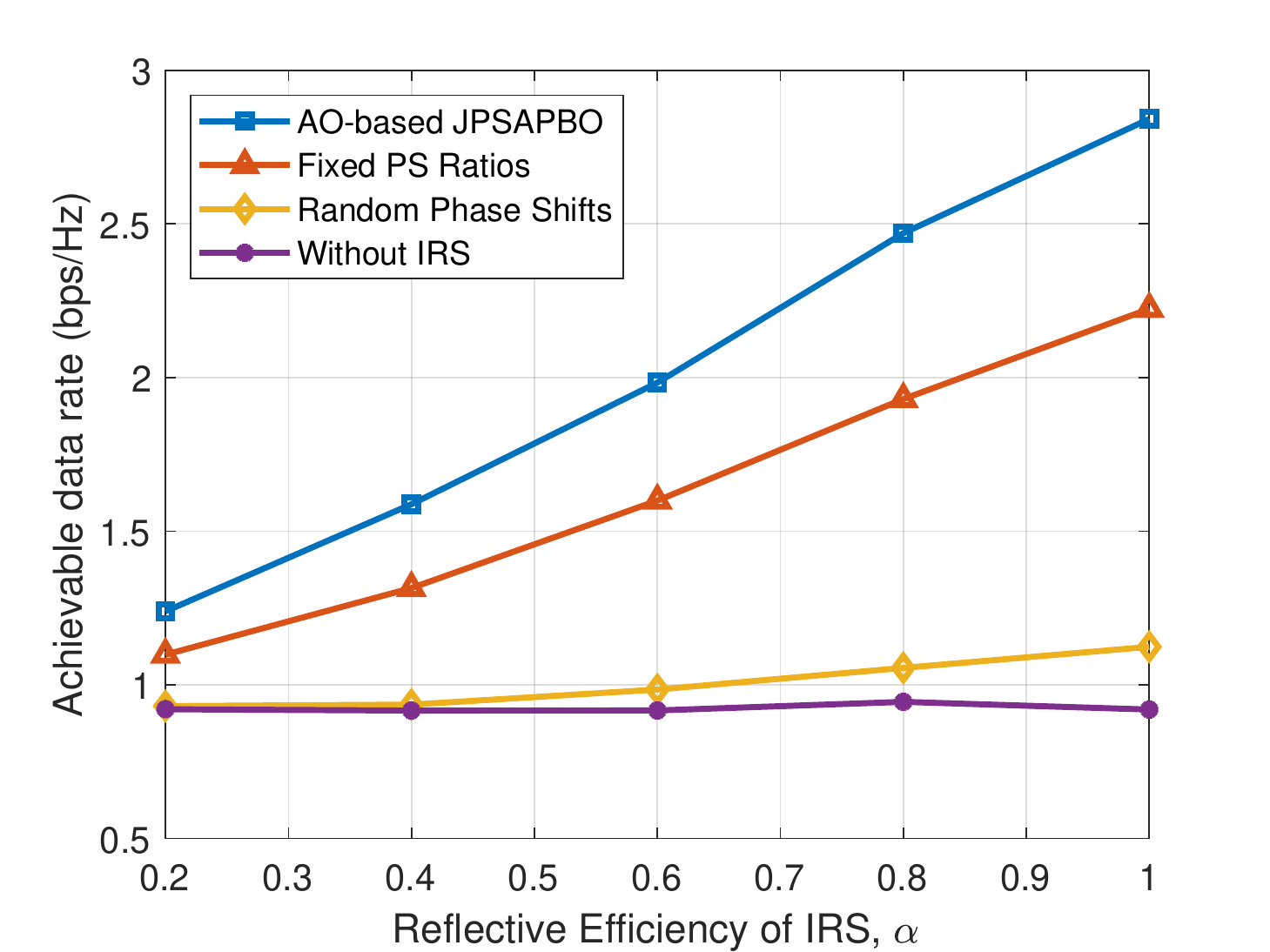}}
    \caption{Achievable data rate versus the reflective efficiency of IRS, i.e., $\alpha$.}
 \label{freflectiveefficiency} 
\end{figure}

As shown in Fig. \ref{freflectiveefficiency}, we study the achievable data rate versus the reflective efficiency of IRS, i.e., $\alpha$. As we can see from the figure, for all IRS reflective efficiencies, the proposed AO-based JPSAPBO algorithm outperforms other benchmarks considerably. Additionally, the Fixed PS Ratios scheme also achieves a better performance gain than both the Random Phase Shifts and the Without IRS schemes. Finally, for the large reflective efficiency of IRS, i.e., $\alpha \ge 0.4$, the Random Phase Shifts scheme achieves a higher data rate compares with the Without IRS scheme, and this implies that the IRS can enhance the performance of the SWIPT system.

\begin{figure}[t]
    \centering
    \scalebox{0.55}{\includegraphics{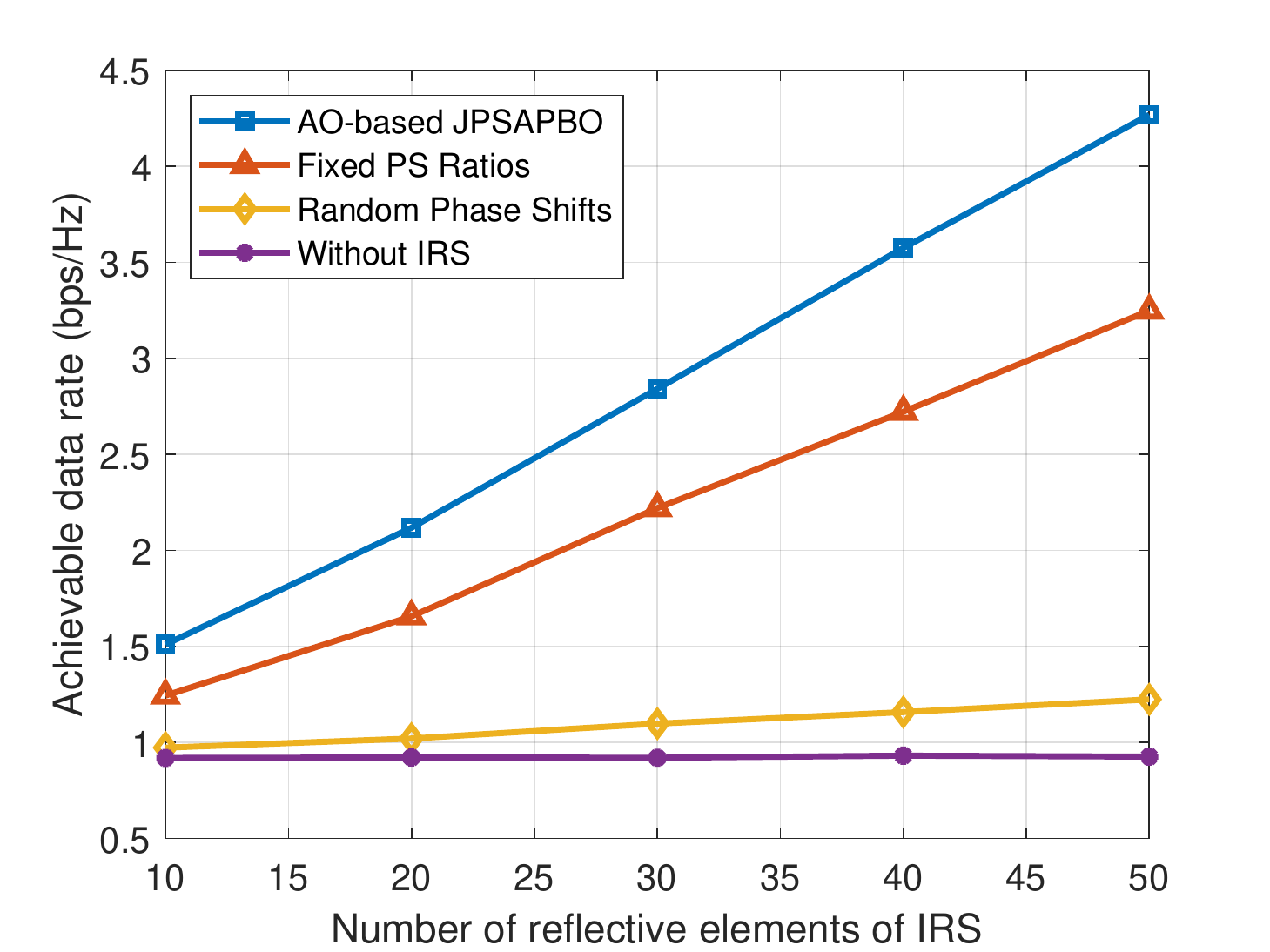}}
    \caption{Achievable data rate versus the number of reflective elements of IRS, i.e., $N$.}
 \label{fnumberirs} 
\end{figure}

As shown in Fig. \ref{fnumberirs}, we study that the achievable data rate versus the number of reflective elements of IRS, i.e., $N$. 
It can be observed that by increasing the number of reflective elements of IRS, the data rate achieved by all schemes except for Without IRS scheme increases monotonically as well, and the proposed AO-based JPSAPBO algorithm outperforms other three benchmarks significantly. In addition, the Fixed PS ratios scheme significantly outperforms both the Random Phase Shifts and the Without IRS schemes, especially for large value of reflective elements, which is due to the fact that PRB of IRS has been carefully optimized. Finally, the Random Phase Shifts scheme also achieves a higher rate than the Without IRS scheme, especially for large number of elements, and this implies that IRS leads to a significant performance gain, even though PRB of IRS has not been optimized.


\section{Conclusion}
In this paper, we investigated the achievable data rate maximization problem in the IRS-assisted SWIPT system with multiple integrated PSRs, which can perform both ID and EH. 
To address this unsolved problem, we proposed a joint optimization framework. Particularly, we decoupled the problem and decomposed it into several sub-problems. Then, we proposed the AO-based JPSAPBO algorithm to solve the sub-problems in an AO manner, and obtain the optimal solutions in nearly closed-forms iteratively.
Simulation results indicated that the proposed  AO-based JPSAPBO algorithm achieved a substantial performance gain and the IRS can significantly enhance the performance.


\appendices
\section{Proof of Lemma}
\subsection{Proof of Lemma \ref{l1}}\label{app1}
Clearly, Problem \ref{eq16} is equivalent to Problem \ref{eq8} with respect to $\rho$, and $\rho$ only exists in the last term of $f_1$. Therefore, we recast $f_1$ as  $f_1\left(\rho,\mathbf W^{\left(t\right)},\Theta^{\left(t\right)},\mathbf {\bar U}^{\left(t+1\right)} \right) \triangleq \operatorname{Tr}\left(\mathcal{F}\left(\rho\right)\right)+ \operatorname{Const}\left(\mathbf {\bar U} \right)$, and where
\begin{align}
    \mathcal F\left(\rho\right)\triangleq\sum_{k=1}^K{\frac{\mathcal A _k}{\mathcal B _k+\frac{\mathcal C_k}{\rho_k}}},
\end{align}
where $\mathcal A _k \triangleq \mathbf {\bar U}_k{\mathbf H_k^{\operatorname{H}}}{\mathbf W_k} \mathbf W_k^{\operatorname{H}}\mathbf H_k$, $\mathcal B _k\triangleq \sum_{i = 1}^K {{\mathbf H_k^{\operatorname{H}}}{\mathbf W_i}\mathbf W_i^{\operatorname{H}}\mathbf H_k}  + \sigma _k^2\mathbf I_{M_u}$ and $\mathcal C _k\triangleq{\delta _k^2}\mathbf I_{M_u}$. 
Since the PS ratio of each PSR is independent, the second-order derivative of $\mathcal F\left(\rho\right)$ as $\frac{{{\partial ^2}\mathcal F\left(\rho\right)}}{{{\partial {{\rho _i}}}{\partial {\rho j}}}} = 0, \; \forall i  \ne j$ holds. 
Therefore, the Hessian matrix of $\mathcal F\left(\rho\right)$ can be expressed as 
\begin{align}
    \frac{\partial ^2 \mathcal F\left(\rho\right)}{\partial \rho ^2} = \operatorname{diag}\left\{\frac{\partial ^2 \mathcal F\left(\rho\right)}{\partial \rho_1 ^2}, \frac{\partial ^2 \mathcal F\left(\rho\right)}{\partial \rho_2 ^2},\cdots,\frac{\partial ^2 \mathcal F\left(\rho\right)}{\partial \rho _K^2}\right\},
\end{align}
each element of the diagonal matrix is determined as 
\begin{align}
    \frac{\partial ^2 \mathcal F\left(\rho\right)}{\partial \rho _k^2}=\frac{-2\mathcal A_k\mathcal B_k\mathcal C_k}{\left(\mathcal B_k\rho_k+\mathcal C_k\right)^3}, \;\forall k.\label{e16}
\end{align}
It is clear that $\mathcal A _k\succeq 0$, $\mathcal B _k\succeq 0$ and $\mathcal C _k\succeq 0, \forall k$ are hold. Therefore, $\frac{\partial ^2 \mathcal F\left(\rho\right)}{\partial \rho _k^2} <0$ is always satisfying, which leads to that $\mathcal F\left(\rho\right)$ is a concave function on $\rho$. 
Moreover, we have the first-order derivative of $\mathcal F\left(\rho\right)$ as $\frac{\partial  \mathcal F\left(\rho\right)}{\partial \rho _k}=\frac{\mathcal A_k \mathcal C_k}{\left(\rho_k \mathcal B_k+\mathcal C_k\right)^2}>0$, $\forall k$. Therefore, $\mathcal F\left(\rho\right)$ is a monotonous increasing function. Similarly, it can be readily verified that the objective function of Problem \ref{eq16} is also concave and monotonously increasing on $\rho$. The proof of the lemma is completed.

\subsection{Proof of Lemma \ref{l2}}\label{app2}
Clearly, the value of $\rho_k$, $\forall k$ should be bounded by $\rho_k^{\min}\le\rho_k\le\rho_k^{\max}$, where $\rho_k^{\min}>0$ ensures for satisfying the minimum value of PS ratio in constraint \eqref{4d}, and $\rho_k^{\max}=1 - \frac{{{\mathbf {\bar e}_k}}}{{{\eta _k}\operatorname{Tr}\left( {\sum_{i = 1}^K {{\mathbf H_k^{\operatorname{H}}}{\mathbf W_i}\mathbf W_i^{\operatorname{H}}\mathbf H_k} } \right)}}$ ensures the harvested energy is larger than the minimum EH requirement in constraint \eqref{4b}. Moreover, according to Lemma \ref{l1}, we know the objective function is monotonous increasing with respect to $\rho$, hence we have 
\begin{align}
    {\rho _k^{\operatorname{\left(t+1\right)}}} \triangleq \max\left\{ \rho_k^{\operatorname{low}}\triangleq \max\left\{ 0,{\rho _k^{\min }}\right\},\rho_k^{\operatorname{up}}\triangleq\min
    \left\{  {\rho _k^{\max },1} \right\} 
    \right\}.
\end{align}
Since $\frac{{{\mathbf {\bar e}_k}}}{{{\eta _k}\operatorname{Tr}\left( {\sum_{i = 1}^K {{\mathbf H_k^{\operatorname{H}}}{\mathbf W_i}\mathbf W_i^{\operatorname{H}}\mathbf H_k} } \right)}}  > 0$ holds, we have $\rho _k^{\operatorname{up} }=\rho _k^{\max}$, $\forall k$. Therefore, when $\rho_k^{\operatorname{up}} > \rho_k^{\operatorname{low}}=0$, Problem \ref{eq16} is feasible. This completes the proof of Lemma \ref{l2}.

\subsection{Proof of Lemma \ref{l3}}\label{app4}
It is clear that the sequence of feasible solutions of ATB $\left\{\mathbf{W}^{\left(n\right)},\forall n=1,2,\cdots,\hat n\right\}$ are the optimal solutions of Problem \ref{eq25}, since they satisfy the KKT conditions of Problem \ref{eq25}. As following, we prove that the converged solution of ATB, i.e., $\mathbf{W}^{\left(\hat n\right)}$ also satisfies the KKT conditions of Problem \ref{eq21}. By defining the corresponding Lagrangian function of Problem \ref{eq21} as following
\begin{align}
    \mathcal{\bar L}\left(\mathbf {W}, \bar\tau, \bar\mu\right)&=\sum_{k=1}^K{\operatorname{Tr}\left(\mathbf {W}_k^{\operatorname{H}}\mathbf {A}
    \mathbf {W}_k\right)} - 2\sum_{k=1}^K{\operatorname{Tr}\left(
    \mathbf{W}_k^{\operatorname{H}}\mathbf {S}_k\right)}\notag\\
    &+\bar\tau\left(\sum_{k=1}^K{\operatorname{Tr} \left(\mathbf {W}_k^{\operatorname{H}} \mathbf {W}_k \right)}-P_{\max}\right)\notag\\
    &-\sum_{k=1}^K{\bar\mu_k\left( \operatorname{Tr}\left (\sum_{i=1}^K{
    {\mathbf {W}_i^{\operatorname{H}}}
    {\mathbf {B}_k}\mathbf {W}_i }\right )-{\mathbf {\dot e} _k}\right)}.
\end{align}
Therefore, the KKT conditions of Problem \ref{eq21} can be represented as 
\begin{align}
    \operatorname{Tr}\left(\left({\mathbf {A} +  \bar\tau \mathbf {I}_{M_b}}\right)\mathbf {W}_k-{\bar\mu_k\mathbf {B}_k}{ \mathbf {W}_k}-
    \mathbf {S}_k \right)&= 0,\forall k;\notag \\
    \bar\tau\left( \sum_{k=1}^K{\operatorname{Tr}\left(\mathbf {W}_k^{H}  \mathbf {W}_k \right)}-P_{\max}\right)&= 0;\notag\\
    \bar\mu_k\left( \operatorname{Tr}\left (\sum_{i=1}^K{
    {\mathbf {W}_i^{\operatorname{H}}}
    {\mathbf {B}_k}\mathbf {W}_i }\right )-{\mathbf {\dot e} _k}\right)&=0,\forall k.
\end{align}
With the converged solution $\mathbf{W}_k^{\left(\hat n\right)},\forall k$, it can be readily checked that there must exist the corresponding Lagrangian multipliers, i.e., $\bar\tau^{\left({\hat n}\right)}$ and $\bar\mu_k^{\left({\hat n}\right)},\forall k$,  for ensuring that the above KKT conditions of Problem \ref{eq21} are satisfied. Hence, the proof of Lemma \ref{l3} is complete.

\subsection{Proof of Lemma \ref{lemma5}}\label{applemma5}
By denoting $\bar\chi_k\ge0,\forall k$ and $\zeta_n\ge0,\forall n$ as the dual variables associated with the constraints \eqref{e46} and \eqref{e48}, the Lagrange function of Problem \ref{eqbeforesca} can be expressed as 
\begin{align}
    \mathcal{L}\left(\theta,\bar\chi, \zeta\right)=f_7\left(\theta\right)&-\sum_{k=1}^K{\bar\chi_k\left(\theta^{\operatorname{H}}\mathbf{\bar{J}}_k\theta+2\Re\left\{\theta^{\operatorname{H}}\lambda_k^*\right\}-{\mathbf {\dddot e}_k}\right)}\notag\\
    &+\sum_{n=1}^N{\zeta_n\left(\left|\theta_n\right|-\frac{1}{\alpha}\right)},
\end{align}
and the KKT conditions are given as 
\begin{align}
    \nabla_{\theta}{f_7\left(\theta\right)}-\sum_{k=1}^K{2\chi_k\left(\theta^{\operatorname{H}}\mathbf{\bar{J}}_k+\lambda_k^*\right)}
    +\sum_{n=1}^N{\zeta_n\nabla_{\theta}{\left(\left|\theta_n\right|\right)}}&=0;\notag\\
    \sum_{k=1}^K{\bar\chi_k\left(\theta^{\operatorname{H}}\mathbf{\bar{J}}_k\theta+2\Re\left\{\theta^{\operatorname{H}}\lambda_k^*\right\}-{\mathbf {\dddot e}_k}\right)}&=0;\notag\\
    {\zeta_n\left(\left|\theta_n\right|-\frac{1}{\alpha}\right)}&=0,\forall n.
\end{align}
We assume that the final converged solution of Algorithm \ref{a2} is $\theta^{\left({\hat q}\right)}$. According to the properties of SCA approach, the constraints \eqref{e46} and \eqref{e47} have the same gradient value at point $\theta=\theta^{\left({\hat q}\right)}$. Meanwhile, according to the properties of MM approach, the same gradient of functions $f_7\left(\theta\right)$ and $g\left(\theta|\theta^{\left({\hat q}\right)}\right)$ achieved at point $\theta^{\left({\hat q}\right)}$. Additional, the constant modulus constraint is ensured by employing the updating rule in \eqref{e54}. Therefore, with the final converged solution of Algorithm \ref{a2}, i.e., $\theta^{\left({\hat q}\right)}$, there must exist the corresponding multipliers $\bar\chi^{\left({\hat q}\right)}$ and $\zeta^{\left({\hat q}\right)}$ for guaranteeing that the KKT conditions of Problem \ref{eqbeforesca} are satisfied. The proof is completed.

%
\section{Connection with MMSE receiver}\label{app3}
We consider the MMSE matrix of each PSR denoted by $\mathbf {\hat L} _k,\forall k$, and the decoded signal can be wittered as  
\begin{align}
    \mathbf{\hat s} _k&=\frac{1}{\sqrt{\rho_k}}\mathbf {\hat L} _k^{\operatorname{H}} \mathbf{y}_k^{\operatorname{ID}}\notag\\
    &=\mathbf {\hat L} _k^{\operatorname{H}}\sum_{i=1}^K{\mathbf {H} _k\mathbf W_i \mathbf s_i}+\mathbf {\hat L} _k^{\operatorname{H}} \mathbf n_k +\frac{1}{\sqrt{\rho_k}}\mathbf {\hat L} _k^{\operatorname{H}} \mathbf z_k,\forall k.
\end{align}
Consequently, the MSE can be expressed as 
\begin{align}
    \operatorname{MSE}_k&=\mathbb{E}\left\{\left(\mathbf {\hat s}_k-\mathbf s_k\right)\left(\mathbf {\hat s}_k-\mathbf s_k\right)^{\operatorname{H}}\right\}\notag\\
    &=\mathbf I_{M_u} + \sum_{i=1}^{K}{\mathbf {\hat L} _k^{\operatorname{H}} \mathbf {H}_k^{\operatorname{H}}\mathbf {W} _i\mathbf {W} _i^{\operatorname{H}}\mathbf {H} _k \mathbf {\hat L} _k} -\mathbf {\hat L} _k^{\operatorname{H}} \mathbf {H}_k^{\operatorname{H}}\mathbf {W} _k\notag\\
    &- \mathbf {W} _k^{\operatorname{H}}\mathbf {H}_k\mathbf {\hat L} _k+\mathbf {\hat L} _k^{\operatorname{H}} \left(\sigma_k^2+\frac{\delta_k^2}{\rho_k}\right) \mathbf {\hat L} _k,\forall k.
\end{align}
Since the maximum transmission power at AP and the minimum EH requirement constraints of PSRs are independent of ID operation, and with fixed ATB and PS ratios, the sub-problem for solving the decoding matrix is given as
\begin{align}
    \mathop{\min}_{\mathbf {\hat L} _k} \operatorname{MSE}_k, \;\forall k.
\end{align}
The closed-form solution of above problem is 
\begin{align}
    {\mathbf {\hat L}_k} = \frac{{\mathbf H_k^{\operatorname{H}}}{\mathbf W_k}}{\sum_{i = 1}^K {{\mathbf H_k^{\operatorname{H}}}{\mathbf W_i}\mathbf W_i^{\operatorname{H}}\mathbf H_k}  + \sigma _k^2\mathbf I_{M_u} + {\delta _k^2}/{\rho _k}\mathbf I_{M_u}}, \;\forall k.
\end{align}

As aforementioned, this appendix implies that there is a connection between the FPT and the MMSE matrices. In fact, MMSE decoding matrix can be interpreted as a special form of the FPT receiver. For details, please refer to \cite{shen2018fractional,shen2018fractional2,8862850}.



\ifCLASSOPTIONcaptionsoff
  \newpage
\fi

\bibliography{ref}
\end{document}